\mag=\magstep1
\documentclass[twoside,reqno]{amsart}
\pdfoutput=1
\usepackage{color}
\usepackage{hyperref}
\usepackage{fancyhdr}
\usepackage{fmtcount}
\usepackage[norelsize,lined,boxed,linesnumbered,algosection]{algorithm2e}
\usepackage{graphicx}
\usepackage[latin1]{inputenc}
\usepackage{amssymb}
\usepackage{tikz}
\usetikzlibrary{matrix}
\usepackage{amsthm}
\usepackage{amsmath}
\usepackage{appendix}
\usepackage{amsfonts}
\usepackage{array}
\usepackage{enumerate}
\usepackage{csvsimple}
\usepackage{booktabs}
\usepackage{longtable}
\usepackage{float}
\usepackage{enumerate}
\usepackage{mathrsfs}
\usepackage{footnote}
\usepackage{tablefootnote}
\usepackage{graphicx}
\usepackage{multirow}
\usepackage{subfig}
\usepackage[all]{hypcap}
\usepackage[font=it]{caption}
\usepackage{myMacro}
\usepackage[protrusion=true,expansion=true]{microtype} 

\SetAlCapSkip{0.5em}

\SetKwInOut{KwInput}{Input}
\SetKwInOut{KwOutput}{Output}

\hypersetup{pdftitle=Actuarial Applications and Estimation of Extended~CreditRisk+, pdfauthor=Jonas Hirz,pdfdisplaydoctitle=true,pdffitwindow=true}

\numberwithin{equation}{section}
\numberwithin{table}{section}
\numberwithin{figure}{section}

\copyrightinfo{}{The authors}

\begin{document}

\pdfpageheight=297 true mm
\pdfpagewidth=210 true mm
\pdfhorigin=1 true mm
\pdfvorigin=4 true mm

\title[Actuarial Applications and Estimation of Extended~CreditRisk$^+$]{Actuarial Applications and Estimation of Extended~CreditRisk$^+$} 

\author[J.~Hirz]{Jonas~Hirz}
\address[J.~Hirz]{BELTIOS GmbH, Lehargasse 1, \mbox{1060 Vienna, Austria}}
\email{\href{mailto:jonas.hirz@beltios.com}{jonas.hirz@beltios.com}}

\author[U.~Schmock]{Uwe~Schmock}
\address[U.~Schmock]{Department of Financial and Actuarial Mathematics, TU Wien, Wiedner Hauptstr.~8--10, \mbox{1040 Vienna, Austria}}
\email{\href{mailto:schmock@fam.tuwien.ac.at}{schmock@fam.tuwien.ac.at}}

\author[P.~V.~Shevchenko]{Pavel~V.~Shevchenko}
\address[Pavel~V.~Shevchenko]{Department of Applied Finance and Actuarial Studies, Macquarie University, NSW, \mbox{2109, Australia}}
\email{\href{pavel.shevchenko@mq.edu.au}{pavel.shevchenko@mq.edu.au}}

\thanks{J.~Hirz gratefully acknowledges financial support from the Australian Government via the 2014 Endeavour Research Fellowship, as well as from the Oesterreichische Nationalbank (Anniversary Fund, project number: 14977) and Arithmetica. 
P.~V.~Shevchenko gratefully acknowledges financial support by the CSIRO-Monash
Superannuation Research Cluster, a collaboration among CSIRO, Monash University, Griffith University, the University of Western Australia, the University of Warwick, and stakeholders of the retirement system in the interest of better outcomes for all.}

\date{\today}

\renewcommand{\subjclassname}{\textup{2\,000} Mathematics Subject
     Classification}

\begin{abstract} 
We introduce an additive stochastic mortality model which allows joint modelling and forecasting of underlying death causes.
Parameter families for mortality trends can be chosen freely. As model settings become high dimensional,  Markov chain Monte Carlo (MCMC) is used for parameter estimation. We then link our proposed model to an extended version of the credit risk model CreditRisk$^+$. This allows exact risk aggregation via an efficient numerically stable Panjer recursion algorithm and provides numerous applications in credit, life insurance and annuity portfolios to derive P\&L distributions.
	Furthermore, the model allows exact (without Monte Carlo simulation error) calculation of
	risk measures and their sensitivities with respect to model parameters for P\&L distributions
	such as value-at-risk and expected shortfall. Numerous examples, including an application to partial internal models under Solvency II, using Austrian and Australian data are shown.
\end{abstract}
	

\maketitle

\tableofcontents

\section{Introduction}
	As the current low interest rate environment forces insurers to put more focus on biometric risks, proper stochastic modelling of mortality has become increasingly important. New regulatory requirements such as Solvency II\footnote{{\href{https://eiopa.europa.eu/regulation-supervision/insurance/solvency-ii}{https://eiopa.europa.eu/regulation-supervision/insurance/solvency-ii}},
 accessed on March 28, 2017.}
 allow the use of internal stochastic models which provide a more risk-sensitive evaluation of
capital requirements and the ability to derive accurate profit and loss (P\&L) attributions with respect to different sources of risk. Benefits for companies which make use of actuarial
tools such as internal models depend crucially on the accuracy of predicted death probabilities and the ability to extract different sources of risk. So far, insurers often use deterministic modelling techniques and then add artificial risk margins to account for risks associated with longevity, size of the portfolio, selection phenomena, estimation and various other sources.
	Such approaches often lack a stochastic foundation and are certainly not consistently appropriate
	for all companies.

	Deriving P\&L distributions of large credit, life and pension portfolios
	typically is a very challenging task.
	In applications, Monte Carlo is the most commonly used approach as it is easy to implement for all different kinds of
	stochastic settings. However, it has shortcomings in finesse and speed, especially for calculation of model sensitivities.
	Motivated by numerical trials, we found that the credit risk model \emph{extended CreditRisk$^+$} (ECRP), as introduced in  \mbox{\cite[section~6]{schmock}}, \mbox{is an exceptionally} efficient as well as
flexible alternative to Monte Carlo and, simultaneously, fits~into life actuarial settings as well.	Coming from credit risk, this model allows flexible handling of	dependence structures within a portfolio via common stochastic risk factors.
 The ECRP model relies on Panjer recursion (cf. \cite{sundt}) which, unlike Monte Carlo, does not require simulation. It~allows an~efficient implementation to derive P\&L distributions exactly given input data and chosen granularity
	associated with discretisation. The speed up for deriving sensitivities, i.e.,~derivatives, of risk measures with respect to model parameters using Panjer recursion will even be an order of magnitude larger as it is extremely difficult to calculate them via finite differencing using Monte Carlo.
	In addition, our~proposed approach
	can enhance pricing of retirement income products and can be used for applications to partial internal models
	in the underwriting risk module.
		
	In Section \ref{stoch_modelling} we introduce an additive stochastic mortality model which is related to classical approaches such as the model introduced in
	\cite{Lee-Carter} or models discussed in ~\cite{cairns2009}. It allows joint modelling of underlying stochastic death causes based on Poisson assumptions where dependence is introduced via common stochastic risk factors.
	Note that forecasting of death causes in a disaggregated way can lead to problems with dominating causes in the long run, as argued in  \cite{wilmoth_1995} as well as  \cite{booth_tickle_2008}. However, joint modelling of death causes can yield computational issues due to high dimensionality which is why the literature is sparse in this matter. An extensive literature review and a multinomial logistic model for joint modelling of death causes is studied in~\cite{Alai_2015}.
	
	Given suitable mortality data, in Section \ref{sec:estimation} we provide several methods to estimate model parameters including matching of
	moments, a maximum a posteriori approach and maximum likelihood as well as Markov chain Monte Carlo (MCMC). Death and population data
	are usually freely available on governmental websites or at statistic bureaus.
	Due to the high dimensionality of our problem, we suggest the use of MCMC which is one of the few statistical approaches allowing joint parameter estimation in high-dimensions. MCMC can become time-consuming, in particular for settings with common stochastic risk factors. However, estimation of model parameters does usually not have to be done on a frequent basis.
	We propose a parameter family for mortality trends which makes our model a generalisation of the Lee--Carter approach. However, our approach
	allows the use of any other kind of parameter family as MCMC is very flexible.
		
		In Section \ref{applications_mortality} we estimate model parameters for Australian death
	data in a setting with 362 model parameters, where trends, trend acceleration/reduction and cohort effects are estimated. Further applications include forecasting of central mortality rates and expected future life time.
	
	In Section \ref{modelling} we then introduce the ECRP model, see  \cite[section 6]{schmock}, which~is a collective risk model corresponding one-to-one to our proposed stochastic mortality model.
	As~the name suggests, it is a credit risk model used to derive loss distributions of credit portfolios
	and originates from the classical CreditRisk$^+$ model which was introduced by~\cite{CRP}.
	Within credit risk models it is classified as a Poisson mixture model.
	Identifying default with death makes the model perfectly applicable for actuarial applications. Extended CreditRisk$^+$ provides a flexible basis for modelling multi-level dependencies and allows a fast and
	numerically stable algorithm for risk aggregation.
	In the ECRP model, deaths are driven by independent stochastic risk~factors.
	The~number of deaths of each policyholder is assumed to be
	Poisson distributed with stochastic intensity.
	Thus, serving as an approximation for the true case with single deaths,
	each person can die multiple times within a period. However, with~proper parameter scaling, approximations are very good and final  loss distributions are accurate
	due to Poisson approximation, as shown in ~\cite{poissonApprox} or~\cite{poisson_approx} as well as the references therein. The close fit of the ECRP model with (mixed) Poisson distributed deaths to more realistic Bernoulli models is outlined in an introductory example. Another great advantage of the ECRP model is that it automatically incorporates many different sources of risks, such as trends, statistical volatility risk and parameter risk.
	
	Section \ref{validation} briefly illustrates validation and model selection techniques.~Model validation approaches are based on previously defined dependence and independence structures.
	All tests suggest that the model suitably fits Australian mortality data.

\section{An Alternative Stochastic Mortality Model}\label{stoch_modelling}

\subsection{Basic Definitions and Notation}
	Following standard actuarial notations and definitions, \cite{pitacco} or ~\cite{cairns2009}, let~$T_{a,g}(t)$ denote the random variable of remaining life
	time of a person aged $a\in\{0,1,\dots, A\}$, with maximum age $A\in \mathbb{N}$, and of gender $g\in\{\textrm{m},\textrm{f}\}$ at time/year $t\in \mathbb{N}$.
	Survival and death probabilities over a time frame $\tau\geq 0$
	are given by ${}_\tau p_{a,g}(t)=\mathbb{P}(T_{a,g}(t)>\tau)$ and
	${}_\tau q_{a,g}(t)=\mathbb{P}(T_{a,g}(t)\leq\tau)$, respectively. For notational purposes we write
	$q_{a,g}(t):={}_1 q_{a,g}(t)$.
	
	Deterministic force of mortality (theory for the stochastic case is also available) at age $a+\tau$ with gender $g$ of a person aged $a$ at time $t$ is given by
	the derivative $\mu_{a+\tau,g}(t):=-\frac{\partial}{\partial \tau}\log {}_\tau p_{a,g}(t)$.
	Henceforth, the central death rate of a person aged $a$ at time $t$ and of gender $g$ is given by
	a weighted average of the force of mortality
	\[
		m_{a,g}(t):=\frac{\int_0^1 {}_s p_{a,g}(t+s)  \mu_{a+s,g}(t+s)\, ds}{\int_0^1 {}_s p_{a,g}(t+s)\, ds}=\frac{q_{a,g}(t)}{\int_0^1 {}_s p_{a,g}(t+s)\, ds}\approx \frac{q_{a,g}(t)}{1-q_{a,g}(t)/2}\,.
	\]

	If $\mu_{a+s,g}(t+s)=\mu_{a,g}(t)$ for all $0\leq s<1$ and $a,t\in\mathbb{N}_0$ with $a\leq A$, i.e.,~under
	piecewise constant force of mortality, we have $m_{a,g}(t)=\mu_{a,g}(t)$ as well as
	$q_{a,g}(t)= 1-\exp(-m_{a,g}(t))$.
	
	Let $N_{a,g}(t)$ denote the number of recorded deaths in year $t$
	of people having age $a$ and gender~$g$, as well as define the exposure to risk $E_{a,g}(t)$
	as the average number of people in year $t$ having age $a$ and gender $g$. The latter
	can often be retrieved from statistical bureaus or approximated by the age-dependent
	population in the middle of a calender year. Estimates for these data in Australia (with~several adjustment components such as census undercount and immigration taken into account) are available at the website of the  {Australian Bureau of Statistics}.
	Considering underlying death causes $k=0,\dots,K$, which are to be understood as diseases or
	injury that initiated the train of morbid events leading directly to death,
	let $N_{a,g,k}(t)$ denote the actual number of recorded deaths due to death cause $k$
	in year $t$ of people having age $a$ and gender $g$. Note that
	$N_{a,g}(t)=N_{a,g,0}(t)+\dots+N_{a,g,K}(t)$. Data on ICD-classified (short for
	International Statistical Classification of Diseases and Related Health Problems)
	death counts can be found for many countries.
	For Australia these data can be found at the
	{Australian Institute of Health and Welfare} (AIHW),
	classified by ICD-9 and ICD-10.
		
	\subsection{Some Classical Stochastic Mortality Models}
	We start with a simple model and assume that deaths in year $t$
	of people having age $a$ and gender $g$ are Poisson distributed
	$N_{a,g}(t)\sim\mathrm{Poisson}(E_{a,g}(t) m_{a,g}(t))$. In this case
	the maximum likelihood estimate for the central death rate is given by
	$\widehat{m}_{a,g}(t)= \widehat N_{a,g}(t)/E_{a,g}(t)$, where $\widehat N_{a,g}(t)$ is the actual recorded number
	of deaths.
	
	The benchmark stochastic mortality model considered in the literature is the traditional Lee--Carter model, \cite{Lee-Carter}, where the logarithmic central death rates are
	modelled in the form
	\[
		\log\widehat{m}_{a,g}(t)=\alpha_{a,g}+\beta_{a,g} \kappa_t+
	\varepsilon_{a,g,t}
	\]
	with independent normal error terms $\varepsilon_{a,g,t}$ with mean 	
	zero, common time-specific component $\kappa_t$, as~well as age and gender specific parameters $\alpha_{a,g}$ and $\beta_{a,g}$. Using suitable
	normalisations, estimates for these parameters and $\kappa_t$ can be derived via the method of moments and singular value decompositions,  ~\cite[section 4.5.1]{AVOE_annuity_table}. Forecasts may then be obtained by applying
auto-regressive models to $\kappa_t$. Note that~\cite{LC_stateSpace2017} and \cite{LC_stateSpace2015} provide joint estimation of parameters and  latent factor $\kappa_t$ in the Lee-Carter model via a state-space framework using MCMC. Various
	extensions of this classical approach with multiple time factors and cohort components have been proposed in the literature; \mbox{for a review}, see~\cite{cairns2009}.

	Instead of modelling central death rates with normal error terms as in the Lee--Carter
	approach, ~\cite{lee-carter1}
	propose to model death counts via Poisson regression where error terms
	are replaced by Poison random variables. In this case
	$N_{a,g}(t)\sim \mathrm{Poisson}(E_{a,g}(t) m_{a,g}(t))$ where, in the simplest case,
	$\log {m}_{a,g}(t)=\alpha_{a,g}+\beta_{a,g} \kappa_t$.
	Correspondingly, assuming that we want to forecast central death rates for different underlying
	death causes $k$, it is natural to assume $N_{a,g,k}(t)\sim\mathrm{Poisson}(E_{a,g}(t) m_{a,g,k}(t))$ where $\log {m}_{a,g,k}(t)=\alpha_{a,g,k}+\beta_{a,g,k} \kappa_{k,t}$.
	However, in this case, it is not difficult to see that
	${m}_{a,g,0}(t)+\dots+{m}_{a,g,K}(t)\neq {m}_{a,g}(t)$, in general, and thus
	\[
		\mathbb{E}[{N_{a,g}(t)}]\neq \sum_{k=0}^K\mathbb{E}[{N_{a,g,k}(t)}]
	\]
	since $N_{a,g}(t)\sim\mathrm{Poisson}\big(E_{a,g}(t)({m}_{a,g,0}(t)+\dots+{m}_{a,g,K}(t))\big)$. Moreover, as central death rates are increasing for selected underlying death cause (e.g.,~central death rates for 75--79 year olds in Australia have doubled from 1987 throughout 2011), forecasts increase exponentially, exceeding one in the~future.
	
	In light of this shortcoming, we will introduce an additive stochastic mortality model which fits into the risk aggregation framework of extended CredtRisk$^+$, see \cite{schmock}.
	
	\subsection{An Additive Stochastic Mortality Model}
	To be able to model different underlying death causes or, more generally, different 	
	covariates which show some common death behaviour (however, we will restrict to the first case in this paper), let~us assume common stochastic risk factors
	$\Lambda_1(t),\dots,\Lambda_K(t)$ with corresponding age-dependent weights
	$w_{a,g,k}(t)$ which give the age-dependent susceptibility to the different risk factors  and which~satisfy
	\[
		w_{a,g,0}(t) + \dots + w_{a,g,K}(t) =1\,.
	\]
	
	\begin{remark}
	Risk factors introduce dependence amongst deaths of different policyholders. If risk factor $\Lambda_k(t)$
		takes large or small values, then the likelihood of death due to $k$ increases or decreases, respectively,
		simultaneously for all policyholders depending on the weight $w_{a,g,k}(t)$. Weights
		$w_{a,g,0},\dots,w_{a,g,K}$
	indicate the vulnerability of people aged $a$ with gender $g$ to risk factors $\Lambda_{1}(t),\dots,\Lambda_{K}(t)$.
	Risk factors are of particular importance to forecast death causes. For a practical example, assume
		that a new, very effective cancer treatment is available such that fewer people die from lung cancer. This situation would have a longevity effect on all policyholders. Such a scenario
		would then correspond to the case when the risk factor for neoplasms shows a small realisation.
	\end{remark}
	
	\begin{definition}[Additive stoch.~mort.~model]\label{def:additiveModel}
		Given risk factors $\Lambda_1(t),\dots,\Lambda_K(t)$ with unit mean and variances $\sigma^2_1(t),\dots,\sigma^2_K(t)$, assume
	\[
		N_{a,g,k}(t)\sim \mathrm{Poisson}\big(E_{a,g}(t) m_{a,g}(t) w_{a,g,k}(t) \Lambda_k(t)\big)\,, \quad k=1,\dots,K\,,
	\]
	being conditionally independent of all $N_{a,g,k'}(t)$ with $k\neq k'$. Idiosyncratic
	deaths $N_{a,g,0}(t)$ with $k=0$ are assumed to be mutually independent and
	independent of all other random variables such that
	\[
		N_{a,g,0}(t)\sim \mathrm{Poisson}\big(E_{a,g}(t) m_{a,g}(t) w_{a,g,0}(t)\big)\,.
	\]
	\end{definition}
	
	In this case, in expectation, deaths due to different underlying death causes add up
	correctly, i.e.,
	\[
		\mathbb{E}[{N_{a,g}(t)}]=E_{a,g}(t) m_{a,g}(t)
		= E_{a,g}(t) m_{a,g}(t) \sum_{k=0}^K w_{a,g,k}(t) = \sum_{k=0}^K \mathbb{E}[{N_{a,g,k}(t)}]
	\]
	as $\mathbb{E}[{\Lambda_k(t)}]=1$ by assumption.
	
	\begin{remark}
	\textls[-15]{In applications, if $K=0$, it is feasible to replace the
	Poisson assumption by a more realistic Binomial assumption $N_{a,g,0}(t)\sim \mathrm{Binomial}\big(E_{a,g}(t), m_{a,g}(t)\big)$, as done in Section \ref{sec:forecaDP} for illustration purposes.}
	\end{remark}
	
	\begin{remark}
		If risk factors are independent and gamma distributed (as in the case of classical CreditRisk$^+$),
		then,~unconditionally, deaths $N_{a,g,k}(t)$ have a negative binomial distribution.
		Then, variance of deaths is given by
		$\var{N_{a,g,k}(t)}=E_{a,g}(t)m_{a,g}(t)  w_{a,g,k}(t)(1+E_{a,g}(t)m_{a,g}(t)  w_{a,g,k}(t)\sigma^2_k(t))$
		with $\sigma^2_k(t)$ denoting the variance of $\Lambda_k(t)$.
		Analogously, for all $a\neq a'$ or $g\neq g'$,
		\begin{equation}\label{covariance_simple}
			\cov(N_{a,g,k}(t),N_{a',g',k}(t))=E_{a,g}(t)\mathbb{E}_{a',g'}(t)m_{a,g}(t) m_{a',g'}(t) w_{a,g,k}(t) w_{a',g',k}(t)\sigma^2_k(t)\,.
		\end{equation}

		This result will be used in Section \ref{validation} for model validation.
	A similar result also holds for the more general model with dependent risk factors, see \cite[section 6.5]{schmock}.
	\end{remark}
	
	To account for improvement in mortality and shifts in death causes over time, we introduce
	the following time-dependent parameter families for trends. Similar to the Lee--Carter model, we could simply consider a linear decrease in log mortality. However, since this yields diminishing or exploding mortality over time, we choose a more sophisticated class with trend reduction features. First, in order to guarantee that values lie in the unit interval, let $F^{\mathrm{Lap}}$ denote the Laplace distribution
	function with mean zero and variance two, i.e.,
	\[
		F^{\mathrm{Lap}}(x)=\frac{1}{2}+\frac{1}{2}\text{sign}(x)\big(1-\exp(-|x|)\big)\,,\quad x\in\mathbb{R}\,,
	\]
	such that, for $x<0$, twice the expression becomes the exponential function.
	
	To ensure that weights and death probabilities are strictly positive for $t\to\infty$, we use the trend reduction/acceleration
	technique
	\begin{equation}\label{CauchyDistr}
		\mathcal{T}^*_{\zeta,\eta}(t)=\frac{1}{\eta}\arctan(\eta(t-\zeta))\,,
	\end{equation}
	\textls[-15]{with parameters $(\zeta,\eta)\in\mathbb{R}\times(0,\infty)$ and $t\in\mathbb{R}$, motivated by~\cite[section 4.6.2]{AVOE_annuity_table}}.
	In particular, (\ref{CauchyDistr}) roughly gives a linear function of $t$ if parameter $\eta$ is small which illustrates the close
	link to the Lee--Carter model. In order to make estimation more stable, we suggest the normalisation $\mathcal{T}_{\zeta,\eta}(t)=(\mathcal{T}^*_{\zeta,\eta}(t)-\mathcal{T}^*_{\zeta,\eta}(t_0))/(\mathcal{T}^*_{\zeta,\eta}(t_0)-\mathcal{T}^*_{\zeta,\eta}(t_0-1))$ with normalisation
	parameter $t_0\in\mathbb{R}$.
	A clear trend reduction in mortality improvements
	can be observed in Japan since 1970, see ~\cite[section 4.2]{DAV}, and also for females in Australia.	
	
	\begin{definition}[Trend families for central death rates and weights]\label{def:families} Central death rates for age $a$, gender $g$ in year $t$ are
	given by
	\begin{equation}\label{eq:PDFamily}
		m_{a,g}(t)=F^{\mathrm{Lap}}\big(\alpha_{a,g} +\beta_{a,g} \mathcal{T}_{\zeta_{a,g},\eta_{a,g}}(t)+\gamma_{t-a}\big)\,,
	\end{equation}
	with parameters \/ $\alpha_{a,g},\beta_{a,g},\zeta_{a,g},\gamma_{t-a}\in\mathbb{R}$
	and\/ $\eta_{i}\in(0,\infty)$, as well as where weights
	are given by
	\begin{equation}\label{eq:WeightFamily}
	w_{a,g,k}(t)= \frac{\exp\big(u_{a,g,k}+v_{a,g,k} \mathcal{T}_{\phi_{k},\psi_{k}}(t)\big)}{\sum_{j=0}^K \exp\big(u_{a,g,j}+v_{a,g,j}\mathcal{T}_{\phi_j,\psi_j}(t)\big)}\,,\quad k\in\{0,\dots,K\}\,,
	\end{equation}
	with parameters \/ $u_{a,g,0},v_{a,g,0},\phi_0,\dots,u_{a,g,K},v_{a,g,K},\phi_K\in\mathbb{R}$ and\/ $\psi_{0},\dots,\psi_{K}\in(0,\infty)$.
	\end{definition}
	
	The assumptions above yield
	an exponential evolution of central death rates over time, modulo trend reduction $\mathcal{T}_{\zeta,\eta}(t)$
	and cohort effects $\gamma_{t-a}$ ($t-a$ refers to the birth year).
	Vector $\alpha$ can be interpreted as intercept parameter for central death rates.
	Henceforth, $\beta$ gives the speed of mortality improvement while $\eta$ gives the speed of trend reduction and $\zeta$
	gives the shift on the S-shaped arctangent curve, i.e.,~the
		location of trend acceleration and trend reduction.
	Parameter $\gamma_{t-a}$ models cohort effects for groups with the same year of birth.
		This factor can also be understood as a categorical variate such as smoker/non-smoker, diabetic/non-diabetic
		or country of residence. The interpretation of model parameters for families of weights is similar.
		
		Cohort effects are not used for modelling weights $w_{a,g,k}(t)$ as sparse data do not allow proper estimation.~In applications, we suggest to fix $\phi$ and $\psi$ in order to reduce dimensionality to suitable~levels.
		Furthermore, fixing trend acceleration/reduction parameters ($\zeta$, $\eta$, $\phi$, $\psi$) yields stable results over time, with similar behavior as in the Lee-Carter model. Including trend reduction parameters can lead to less stable results over time. However, our proposed model allows free adaption of parameter families for mortality and weights.

\begin{remark}[Long-term projections]\label{rem:long-term}
	Long-term projections of death probabilities using (\ref{eq:PDFamily})
	give
	\[
		\lim_{t\to\infty} m_{a,g}(t)=
			F^{\mathrm{Lap}}\Big(\alpha_{a,g} +\beta_{a,g} \frac{\pi}{2 \eta_{a,g}}\Big)\,.
	\]

	\noindent Likewise, long-term projections for weights
	using (\ref{eq:WeightFamily}) are given by
	\[
		\lim_{t\to\infty} w_{a,g,k}(t)=
			\frac{\exp\big(u_{a,g,k} +v_{a,g,k}\frac{\pi}{2 \psi_{k}}\big)}{\sum_{j=0}^K \exp\big(u_{a,g,j} +v_{a,g,j} \frac{\pi}{2 \psi_{j}}\big)}\,.
	\]

	\noindent Thus, given weak trend reduction, i.e.,~$\psi_k$ close to zero,
	weights with the strongest trend will tend to dominate in the long term. If we a priori fix the parameter for trend reduction $\psi_k$ at suitable values, this effect can be controlled. Alternatively, different parameter families for weights can be used, e.g.,~linear families. Note~that our model ensures that weights across risk factors $k=0,1,\dots,K$ always sum up to one which is why overall mortality $m_{a,g}(t)$ is not influenced by weights and their trends.
\end{remark}

\section{Parameter Estimation}\label{sec:estimation}

 In this section we provide several approaches for parameter estimation in our proposed model from Definitions \ref{def:additiveModel} and \ref{def:families}. The approaches include
maximum likelihood, maximum a posteriori, matching of moments and MCMC. Whilst matching of moments estimates are easy to derive but less accurate,
maximum a posterior and maximum likelihood estimates cannot be calculated by deterministic numerical optimisation, in general. Thus, we suggest MCMC as a slower but very powerful alternative. Publicly available data based on
the whole population of a country are used.

~\cite{embrechts} in section 8.6 consider statistical inference for Poisson mixture models and
Bernoulli mixture models. They briefly introduce moment estimators and maximum likelihood estimators
for homogeneous groups in Bernoulli mixture models. Alternatively, they derive statistical inference via
a generalised linear mixed model representation for mixture models which is
distantly related to our setting. In their \lq Notes and Comments\rq~section the reader can find a comprehensive list of
interesting references. Nevertheless, most of their results and arguments are not directly applicable to our case
since we use a different parametrisation and since we usually have rich data of death counts compared to the sparse data on company defaults.

In order to be able to derive statistically sound estimates, we make the following simplifying assumption for time independence:

\begin{definition}[Time independence and risk factors]\label{simple_assumptions}
	Given Definition \ref{def:additiveModel}, consider discrete-time periods\/
	$U:=\{1,\dots,T\}$ and assume that random variables are independent for different points in time\/ $s\neq t$  in $U$. Moreover, for each $t\in U$, risk factors $\Lambda_1(t),\dots,\Lambda_K(t)$ are assumed to be independent and, for each\/ $k\in\{1,\dots,K\}$, \/ $\Lambda_k(1),\dots,\Lambda_k(T)$ are identically gamma distributed with mean one and variance $\sigma_k^2\geq 0$.
\end{definition}

The assumptions made above seem suitable for Austrian and Australian data, as shown in Section~\ref{validation} via model validation.
In particular, serial dependence is mainly captured by trend families in death probabilities and weights.

For estimation of life tables we usually assume $K=0$ or $K=1$ with $w_{a,g,1}=1$ for all ages and~genders. For~estimation and forecasting of
death causes, we  identify risk factors with underlying death causes.
Note that for fixed $a$ and $g$, Equation  (\ref{eq:WeightFamily}) is invariant under a
	constant shift of parameters $(u_{\mathrm{a},\mathrm{g},k})_{k\in\{0,\dots,K\}}$ as
	well as $(v_{\mathrm{a},\mathrm{g},k})_{k\in\{0,\dots,K\}}$
	if $\phi_0=\dots=\phi_K$ and $\psi_0=\dots=\psi_K$, respectively. Thus,
for each $a$ and $g$,
we can always choose fixed and arbitrary values for $u_{a,g,0}$
and $v_{a,g,0}$.

\subsection{Estimation via Maximum Likelihood}\label{max_like}
We start with the classical \emph{maximum likelihood} approach. The likelihood function can be derived in closed form but, unfortunately, estimates have to be derived via MCMC as deterministic numerical optimisation
quickly breaks down due to high dimensionality.

\begin{lemma}[Likelihood function]\label{MLE}
	Given Definitions\/ \ref{def:additiveModel}--\ref{simple_assumptions},
	define
	\[
		\widehat N_k(t):=\sum_{a=0}^A\sum_{g\in\{\mathrm{f},\mathrm{m}\}}\widehat N_{a,g,k}(t)\,,
		\quad k\in\{0,\dots,K\}\textrm{ and }t\in U\,,
	\]
	as well as\/ $\rho_{a,g,k}(t):=E_{a,g}(t) m_{a,g}(t) w_{a,g,k}(t)$ for all age groups\/ $a$, with maximum age group $A$, and gender\/ $g$~and
	\[
		\rho_k(t):=\sum_{a=0}^A\sum_{g\in\{\mathrm{f},\mathrm{m}\}}\rho_{a,g,k}(t)\,.
	\]

	Then, the likelihood function $\ell( \widehat N|\theta_m,\theta_w,{\sigma})$
	of parameters\/ $\theta_m:=({\alpha},{\beta},\zeta,\eta)\in E$, as well as\/ $\theta_w:=({u},{v},\phi,\psi)\in F$ and\/ ${\sigma}:=(\sigma_k)\in[0,\infty)^{K}$
	 given\/
	$ \widehat N:=(\widehat N_{a,g,k}(t))\in\mathbb{N}_0^{A\times 2\times (K+1)\times T}$
	is given~by
	\begin{equation}\label{MLE_likelihood}
	\begin{split}
		\ell&(\widehat N|\theta_m,\theta_w,{\sigma})=
		\prod_{t=1}^T \Bigg(\bigg(\prod_{a=0}^A\prod_{g\in\{\mathrm{f},\mathrm{m}\}}\frac{e^{-\rho_{a,g,0}(t)} \rho_{a,g,0}(t)^{\widehat N_{a,g,0}(t)}}
		{\widehat N_{a,g,k}(t)!}\bigg)\\
		&\times		
		\prod_{k=1}^K\bigg(\frac{\Gamma(\sigma^{-2}_k+\widehat N_k(t))}
		{\Gamma(\sigma^{-2}_k) \sigma^{2\sigma^{-2}_k}_k (\sigma^{-2}_k+\rho_k(t))^{\sigma^{-2}_k+\widehat N_k(t)}}\prod_{a=0}^A\prod_{g\in\{\mathrm{f},\mathrm{m}\}}
		\frac{ \rho_{a,g,k}(t)^{\widehat N_{a,g,k}(t)}}
		{ \widehat N_{a,g,k}(t)!}\bigg)\Bigg)\,.
	\end{split}
	\end{equation}
\end{lemma}

\begin{proof}
Following our assumptions, by straightforward computation {we get}
\[
	\begin{split}
		\ell(\widehat N|\theta_m,\theta_w,{\sigma})=
		\prod_{t=1}^T \Bigg(&\bigg(\prod_{a=0}^A\prod_{g\in\{\mathrm{f},\mathrm{m}\}}\frac{ e^{-\rho_{a,g,0}(t)}\rho_{a,g,0}(t)^{\widehat N_{a,g,0}(t)}}
		{\widehat N_{a,g,0}(t)!}\bigg)\\
		&\hspace{-102pt}\times\prod_{k=1}^K
		\mathbb{E}\Bigg[{\mathbb{P}\bigg(\bigcap_{a=0}^A\bigcap_{g\in\{\mathrm{f},\mathrm{m}\}}
		\bigl\lbrace N_{a,g,k}(t)= \widehat N_{a,g,k}(t)\bigr\rbrace\,\bigg|\,\Lambda_k(t)\bigg)}\Bigg]\Bigg)\,,
	\end{split}
\]
where $\ell(\widehat N|\theta_m,\theta_w,{\sigma})=\mathbb{P}({N}=\widehat{N}|\theta_m,\theta_w,{\sigma})$ denotes the
probability of the event $\{{N}=\widehat{N}\}$ given parameters.
Taking expectations in the equation above yields
\[
	\begin{split}
		\mathbb{E}&{\bigg[\mathbb{P}\bigg(\bigcap_{a=0}^A\bigcap_{g\in\{\mathrm{f},\mathrm{m}\}}
		\bigl\lbrace N_{a,g,k}(t)=\widehat N_{a,g,k}(t)\bigr\rbrace\,\bigg|\,\Lambda_k(t)\bigg)}\bigg]\\
		&=\bigg(\prod_{a=0}^A\prod_{g\in\{\mathrm{f},\mathrm{m}\}}\frac{\rho_{a,g,k}(t)^{\widehat N_{a,g,k}(t)}}
		{\widehat N_{a,g,k}(t)!}\bigg)\int_0^\infty
		e^{-\rho_k(t) x_t}x_t^{\widehat N_{k}(t)} \frac{x_t^{\sigma^{-2}_k-1} e^{-x_t\sigma^{-2}_k}}{\Gamma(\sigma^{-2}_k) \sigma^{2\sigma^{-2}_k}_k}\,dx_t\,.
	\end{split}
\]

The integrand above is a density of a gamma distribution---modulo the normalisation constant---with parameters
$\sigma^{-2}_k+\widehat N_k(t)$ and $\sigma^{-2}_k+\rho_k(t)$. Therefore, the corresponding integral equals the multiplicative inverse of the normalisation constant, i.e.,

\[
\begin{split}
	\bigg(\frac{(\sigma^{-2}_k+\rho_k(t))^{\sigma^{-2}_k+\widehat N_k(t)}}{\Gamma(\sigma^{-2}_k+\widehat N_k(t))}\bigg)^{-1}\,,\quad
	k\in\{1,\dots,K\}\textrm{ and }t\in\{1,\dots,T \}\,.
\end{split}
\]

Putting all results together gives (\ref{MLE_likelihood}).
\end{proof}

Since the products in (\ref{MLE_likelihood}) can become small, we recommend to use the log-likelihood function instead.
For implementations we recommend to use the log-gamma function, e.g.,~the {lgamma} function in  `R'~see \cite{stats}.

\begin{definition}[Maximum likelihood estimates]\label{def_MLE}
	Recalling (\ref{MLE_likelihood}), as well as
	given the assumptions of Lemma \ref{MLE},  \emph{maximum likelihood estimates} for parameters
	$\theta_m,\theta_w$ and $\sigma$
		are defined by
		\end{definition}
\vspace{-12pt}

	\[
		\big({\hat{\theta}_m}^{\mathrm{MLE}},{\hat{\theta}_w}^{\textrm{MLE}},{\hat{\sigma}}^{\textrm{MLE}}\big):=\mathop{\text{arg\,sup}}_{\theta_m,\theta_w,{\sigma}}
		\ell({\widehat N}|\theta_m,\theta_w,{\sigma})=\mathop{\text{arg\,sup}}_{\theta_m,\theta_w,{\sigma}}\log\ell(\widehat{N}|\theta_m,\theta_w,{\sigma})\,.
	\]

	Deterministic optimisation of the likelihood function may quickly lead to
	numerical issues due to high dimensionality.  In \lq R\rq~the deterministic optimisation routine \textsf{nlminb}, see  \cite{stats},
	gives stable results in simple examples.
	Our proposed alternative is to switch to a Bayesian setting and use MCMC as described in
	Section \ref{sec:MCMC}.

\subsection{Estimation via a Maximum a Posteriori Approach}\label{sec:MAP}

Secondly we propose a variation of \emph{maximum a posteriori estimation} based on
Bayesian inference, \cite[section 2.9]{pavel}. If risk factors are not integrated out in the likelihood
function, we may also derive the posterior density of the risk factors as follows.
One main advantage of this approach is that estimates for risk factors are obtained
which is very useful for scenario analysis and model validation. Furthermore, handy
approximations for estimates of risk factor realisations and variances are obtained.

\begin{lemma}[Posterior density]\label{posteriori}
	Given Definitions\/ \ref{def:additiveModel}--\ref{simple_assumptions}, consider parameters\/ $\theta_m:=({\alpha},{\beta},\zeta,\eta)\in E$,
	$\theta_w:=({u},{v},\phi,\psi)\in F$, as well as
	realisations\/ $\lambda:=(\lambda_k(t))\in(0,\infty)^{K\times T}$ of risk factors\/
	${\Lambda}:=(\Lambda_k(t))\in(0,\infty)^{K\times T}$, as well as data\/
	$\widehat N:=(\widehat N_{a,g,k}(t))\in\mathbb{N}_0^{A\times 2\times (K+1)\times T}$.
	Assume that their prior distribution is denoted by $\pi(\theta_m,\theta_w,{\sigma})$.
	Then,
	the posterior density\/ $\pi(\theta_m,\theta_w,{\lambda},{\sigma}| \widehat{N})$
	of parameters given data\/
	$ \widehat{N}$ is up to constant given~by
	\begin{equation}\label{likelihood_simple}
		\begin{split}
			\pi&(\theta_m,\theta_w,{\lambda},{\sigma}| \widehat{N})\propto
			  \pi(\theta_m,\theta_w,{\sigma}) \pi({\lambda}|\theta_m,\theta_w,{\sigma}) \ell( \widehat{N}|\theta_m,\theta_w,{\lambda},{\sigma})\\
			&=\prod_{t=1}^T\Bigg(\bigg(\prod_{a=0}^A\prod_{g\in\{\mathrm{f},\mathrm{m}\}}\frac{e^{-\rho_{a,g,0}(t)}\rho_{a,g,0}(t)^{\widehat N_{a,g,0}(t)}}{\widehat N_{a,g,0}(t)!}\bigg)
			\prod_{k=1}^K\bigg(
			\frac{e^{-\lambda_k(t)\sigma^{-2}_k }\lambda_k(t)^{\sigma^{-2}_k-1}}{\Gamma(\sigma^{-2}_k) \sigma^{2 \sigma^{-2}_k}_k}\\
			& \times\prod_{a=0}^A\prod_{g\in\{\mathrm{f},\mathrm{m}\}}
			\frac{e^{-\rho_{a,g,k}(t) \lambda_k(t)}(\rho_{a,g,k}(t) \lambda_k(t))^{\widehat N_{a,g,k}(t)}}{\widehat N_{a,g,k}(t)!}
			\bigg)\Bigg)   \pi(\theta_m,\theta_w,{\sigma})\,,
		\end{split}
	\end{equation}
	where\/
	$\pi({\lambda}|\theta_m,\theta_w,{\sigma})$
	denotes the prior distribution of risk factors at\/ ${\Lambda}={\lambda}$ given all other parameters, where\/
	$\ell( \widehat{N}|\theta_m,\theta_w,{\lambda},{\sigma})$ denotes the likelihood of\/ ${N}=\widehat{N}$ given all parameters and
	 where\/ $\rho_{a,g,k}(t)=E_{a,g}(t) m_{a,g}(t) w_{a,g,k}(t)$.
\end{lemma}

\begin{proof} The first proportional equality follows by Bayes' theorem which is also widely used in
	Bayesian inference, see, for example, \cite[section 2.9]{pavel}.
	Moreover,
	\[
		\pi({\lambda}|\theta_m,\theta_w,{\sigma})=\prod_{k=1}^K\prod_{t=1}^T
		\bigg(\frac{e^{-\lambda_k(t)\sigma^{-2}_k}\lambda_k(t)^{\sigma^{-2}_k-1}}{\Gamma(\sigma^{-2}_k) \sigma^{2 \sigma^{-2}_k}_k}
		\bigg)\,.
	\]

	If $\theta_m\in E$, $\theta_w\in F$, ${\lambda}\in(0,\infty)^{K\times T}$ and ${\sigma}\in[0,\infty)^{K}$, then note that
	\[
		\begin{split}
		\ell( \widehat{N}|\theta_m,\theta_w,{\lambda},{\sigma})&=
			\prod_{a=0}^A\prod_{g\in\{\mathrm{f},\mathrm{m}\}}\prod_{t=1}^T\bigg(e^{-\rho_{a,g,0}(t)}
			\frac{\rho_{a,g,0}(t)^{\widehat N_{a,g,0}(t)}}{\widehat N_{a,g,0}(t)!}\\				&\times\prod_{k=1}^K\mathbb{P}\big(
			 N_{a,g,k}(t)=\widehat N_{a,g,k}(t)\,\big|\,		
			\Lambda_k(t)=\lambda_k(t)\big)\bigg)\,,
		\end{split}
	\]
	which then gives 	(\ref{likelihood_simple}) by straightforward computation as in Lemma \ref{MLE}.
\end{proof}

The approach described above may look like a pure Bayesian inference approach but note that
risk factors $\Lambda_k(t)$ are truly stochastic and, therefore, we refer to it as a maximum a posteriori
estimation approach. There are many reasonable choices for prior distributions of parameters which include (improper) uniform priors
$\pi(\theta_m,\theta_w,{\sigma}):=
		1_E({\theta_m})1_{F}({\theta_w}) 1_{(0,\infty)^K}({\sigma})$ to smoothing priors
		as given in Section \ref{sec:forecaDP}.
Having derived the posterior density, we can now define corresponding maximum a~posteriori estimates.

\begin{definition}[Maximum a posteriori estimates]\label{def_posteriori}
	Recalling (\ref{likelihood_simple}), as
	well as given the assumptions of Lemma~\ref{posteriori}, \emph{maximum a posteriori estimates} for parameters
	$\theta_m,\theta_w,{\lambda}$ and ${\sigma}$, given uniqueness, are defined by
	\[
		\begin{split}
		\big({\hat{\theta}_m}^{\mathrm{MAP}},{\hat{\theta}_w}^{\mathrm{MAP}},{\hat{\lambda}}^{\mathrm{MAP}},{\hat{\sigma}}^{\mathrm{MAP}}\big) & :=
		\mathop{\emph{arg\,sup}}_{\theta_m,\theta_w,{\lambda},{\sigma}}\pi(\theta_m,\theta_w,{\lambda},{\sigma}| \widehat{N})\\ & \,\,=\mathop{\emph{arg\,sup}}_{\theta_m,\theta_w,{\lambda},{\sigma}}\log\pi(\theta_m,\theta_w,{\lambda},{\sigma}|\widehat{N})\,.
		\end{split}
	\]
\end{definition}

Again, deterministic optimisation of the posterior function may quickly lead to
	numerical issues due to high dimensionality
	of the posterior function which is why we recommend MCMC.
However, we can provide handy approximations for risk factor and variance estimates.

\begin{lemma}[Conditions for maximum a posteriori estimates]\label{prop_MAP}
	Given Definition \ref{def_posteriori},  estimates\/
	${\hat{\lambda}}^{\mathrm{MAP}}$  and\/
	${\hat{\sigma}}^{\mathrm{MAP}}$
	satisfy,  for every\/ $k\in\{1,\dots,K\}$ and\/ $t\in U$,
	\begin{equation}\label{lambda_MAP}
		\hat{\lambda}^{\mathrm{MAP}}_k(t)=\frac{(\hat{\sigma}_k^{\mathrm{MAP}})^{-2}-1+\sum_{a=0}^A\sum_{g\in\{\mathrm{f},\mathrm{m}\}}\widehat N_{a,g,k}(t)}
	{(\hat{\sigma}_k^{\mathrm{MAP}})^{-2}+\sum_{a=0}^A\sum_{g\in\{\mathrm{f},\mathrm{m}\}} \rho_{a,g,k}(t)}
	\end{equation}
	if\/ $(\hat{\sigma}_k^{\mathrm{MAP}})^{-2}-1+\sum_{a=0}^A\sum_{g\in\{\mathrm{f},\mathrm{m}\}} \widehat N_{a,g,k}(t)> 0$, as well as
	\begin{equation}\label{beta_MAP}
		2\log \hat{\sigma}_k^{\mathrm{MAP}}+\frac{\Gamma'\big((\hat{\sigma}_k^{\mathrm{MAP}})^{-2}\big)}{\Gamma\big((\hat{\sigma}_k^{\mathrm{MAP}})^{-2}\big)}
		=\frac{1}{T}\sum_{t=1}^T\big(1+\log\hat{\lambda}^{\mathrm{MAP}}_k(t)-\hat{\lambda}^{\mathrm{MAP}}_k(t)\big)\,,
	\end{equation}
	where, for given\/ $\hat{\lambda}^{\mathrm{MAP}}_k(1),\dots,\hat{\lambda}^{\mathrm{MAP}}_k(T)>0$, (\ref{beta_MAP})
	has a unique solution which is strictly positive.
\end{lemma}

\begin{proof}
	First, set $\pi^*(\widehat{N})
		:=\log\pi(\theta_m,\theta_w,{\lambda},{\sigma}|\widehat{N})$.
	Then, for every $k\in\{1,\dots,K\}$ and $t\in U$, differentiating $\pi^*(\widehat{N})$ gives
	\[
		\frac{\partial\pi^*( \widehat{N})}{\partial \lambda_k(t)}
		=\frac{\sigma^{-2}_k-1}{\lambda_k(t)}-\frac{1}{\sigma^2_k}+\sum_{a=0}^A\sum_{g\in\{\mathrm{f},\mathrm{m}\}}\Big(\frac{\widehat N_{a,g,k}(t)}{\lambda_k(t)}-\rho_{a,g,k}(t)\Big)\,.
	\]

	Setting this term equal to zero and solving for $\Lambda_k(t)$ gives (\ref{lambda_MAP}).
	Similarly, for every $k\in\{1,\dots,K\}$, we obtain
	\[
		\frac{\partial\pi^*(\widehat{N})}{\partial\sigma^2_k}
		=\frac{1}{\sigma_k^4}\sum_{t=1}^T\Big(\log\sigma^2_k-1+\frac{\Gamma'(\sigma^{-2}_k)}{\Gamma(\sigma^{-2}_k)}-\log\lambda_k(t)+\lambda_k(t)\Big)\,.
	\]

	Again, setting this term equal to zero and rearranging the terms gives (\ref{beta_MAP}).
			
	For existence, uniqueness of the solution in (\ref{beta_MAP}), let
	$\hat{\lambda}^{\mathrm{MAP}}_k(1),\dots,\hat{\lambda}^{\mathrm{MAP}}_k(T)>0$ and assume $k\in\{1,\dots,K\}$ to be fixed.  Then,
	note that the right side in (\ref{beta_MAP}) is strictly negative
	unless $\hat{\lambda}^{\mathrm{MAP}}_k(1)=\dots=\hat{\lambda}^{\mathrm{MAP}}_k(T)=1$, as $\log x \leq x-1$ for all $x>0$ with equality for $x=1$. If~$\hat{\lambda}^{\mathrm{MAP}}_k(1)=\dots=\hat{\lambda}^{\mathrm{MAP}}_k(T)=1$, then there is no variability in the risk factor such that $\sigma^2_k=0$. Henceforth, note that
	$f(x):=\log x-{\Gamma'(x)}/{\Gamma(x)}$, for all $x>0$,
	is continuous ($\Gamma'(x)/\Gamma(x)$ is known as digamma function or $\psi$-function) with
	\begin{equation}\label{digamma_approx}
		\frac{1}{2 x}<f(x)<\frac{1}{2 x}+\frac{1}{12 x^2}\,,\quad x>0\,,
	\end{equation}
	which follows by~\cite[Corollary 1]{digamma2} and
	$f(x+1)=1/x+f(x)$ for all $x>0$. As we want to solve $-f(1/x)=-c$ for some given $c>0$,
	note that $f(0+)=\infty$, as well as
	$\lim_{x\to\infty}f(x)=0$. \mbox{Thus, a solution} has to exist as $f(1/x)$ is continuous on $x>0$. Furthermore,
	\[
		f'(x)=\frac{1}{x}-\sum_{i=0}^\infty\frac{1}{(x+i)^2}< \frac{1}{x}-\int_{x}^\infty \frac{1}{z^2}\,dz= 0
		\,,\quad x>0\,,
	\]
	where the first equality follows by \cite{digamma}. This implies that  $f(x)$ and $(-f(1/x))$ are strictly decreasing. Thus, the solution in (\ref{beta_MAP}) is unique.
\end{proof}

Using Lemma \ref{prop_MAP}, it is possible to derive handy approximations for risk factor and variance estimates, given  estimates for weights and death probabilities which can be derived by matching of moments as given in Section \ref{sec:MM} or other models such as Lee--Carter. If $\sum_{a=0}^A\sum_{g\in\{\mathrm{f},\mathrm{m}\}}\widehat N_{a,g,k}(t)$
	is large, it is reasonable to define
	\begin{equation}\label{MAPappr_lambda}
		\hat{\lambda}^{\mathrm{MAPappr}}_k(t):=\frac{-1+\sum_{a=0}^A\sum_{g\in\{\mathrm{f},\mathrm{m}\}}\widehat N_{a,g,k}(t)}
		{\sum_{a=0}^A\sum_{g\in\{\mathrm{f},\mathrm{m}\}} \rho_{a,g,k}(t)}
	\end{equation}
	as an approximative estimate for $\lambda_k(t)$
	where $\rho_{a,g,k}(t):=E_{a,g}(t) m_{a,g}(t) w_{a,g,k}(t)$.
	Having derived approximations for ${\lambda}$,
	we can use (\ref{beta_MAP})
	to get estimates for ${\sigma}$.
	Alternatively, note that due to (\ref{digamma_approx}), we get
	\[
		-2\log \hat{\sigma}_k^{\mathrm{MAP}}-\frac{\Gamma'\big((\hat{\sigma}_k^{\mathrm{MAP}})^{-2}\big)}{\Gamma\big((\hat{\sigma}_k^{\mathrm{MAP}})^{-2}\big)}
		= \frac{(\hat{\sigma}_k^{\mathrm{MAP}})^2}{2}+\mathcal{O}\big((\hat{\sigma}_k^{\mathrm{MAP}})^4\big)\,.
	\]

	Furthermore, if we use second order Taylor expansion for the logarithm,
	then the right hand side of
	 (\ref{beta_MAP}) gets
	\[
		\frac{1}{T}\sum_{t=1}^T\big(\hat{\lambda}^{\mathrm{MAP}}_k(t)-1-\log\hat{\lambda}^{\mathrm{MAP}}_k(t)\big)=
		\frac{1}{2 T}\sum_{t=1}^T\Big( \big(\hat{\lambda}^{\mathrm{MAP}}_k(t)-1\big)^2+
		\mathcal{O}\big(\big(\hat{\lambda}^{\mathrm{MAP}}_k(t)-1\big)^3\big)\Big) \,.
	\]

	This approximation is better the closer the values of ${\lambda}$ are  to one.
	Thus, using these observations, an approximation for risk factor variances ${\sigma}^2$ is given by
	\begin{equation}\label{MAPappr_beta}
		 \big(\hat{\sigma}_k^{\mathrm{MAPappr}}\big)^2:={
		\frac{1}{T}\sum_{t=1}^T \big(\hat{\lambda}^{\mathrm{MAPappr}}_k(t)-1\big)^2}\,,
	\end{equation}
	which is the sample variance of ${\hat{\lambda}}^{\mathrm{MAP}}$.
	Note $|\hat{\lambda}^{\mathrm{MAP}}_k(t)-1|<|\hat{\lambda}^{\mathrm{MAPappr}}_k(t)-1|$,
	implying that (\ref{MAPappr_beta}) will dominate
	solutions obtained by (\ref{beta_MAP}) in most cases.

\subsection{Estimation via MCMC}\label{sec:MCMC}
As we have already outlined in in the previous sections, deriving maximum a posteriori estimates and
maximum likelihood estimates via deterministic numerical optimisation is mostly impossible due to high dimensionality
(several hundred parameters). Alternatively, we can use {MCMC} under a~Bayesian setting.
Introductions to this topic can be found, for example, in~\cite{MCMC},
\cite{gamerman}, as well as \cite[section~2.11]{pavel}.
We suggest to use the random walk Metropolis--Hastings within Gibbs algorithm which,
given that the Markov chain is irreducible and aperiodic,
generates sample chains that converge to the stationary distribution,
 \cite{tierney} and also \cite[sections 6--10]{Roberts2004}. However, note that
various MCMC algorithms are available.

MCMC requires a Bayesian setting which we
automatically have in the maximum a posteriori approach, see Section \ref{sec:MAP}. Similarly, we can switch
to a Bayesian setting in the maximum likelihood approach, see Section \ref{max_like}, by simply multiplying the
likelihood function with a prior distribution of parameters.
MCMC generates Markov chains which provide samples from the posterior distribution where the mode of these samples then corresponds
to an approximation for the maximum a posteriori estimate. More stable estimates in terms of mean squared error
are obtained by taking the mean over all samples once MCMC chains sample
from the stationary distribution, \mbox{\cite[section 2.10]{pavel}}. Taking the mean over all samples as an estimate, of course, can lead to troubles if posterior
distributions of parameters are, e.g., bimodal, such that we end up in a region which is highly unlikely.
Furthermore, sampled posterior distribution can be used to estimate parameter uncertainty.
The method requires a~certain burn-in period until the generated chain becomes stationary. Typically, one tries to get average acceptance probabilities close to $0.234$
which is asymptotically optimal for multivariate Gaussian proposals as shown in~\cite{optimal_accept}.
To reduce long computational times, one can run several independent MCMC chains with different starting points on different
CPUs in a parallel way. To~prevent overfitting, it is possible to regularise, i.e., smooth, maximum a posteriori estimates via
adjusting the prior distribution. This technique is particularly used in regression, as well as in many applications, such as
signal processing. When forecasting death probabilities in Section \ref{sec:forecaDP}, we use a~Gaussian prior density with a certain
correlation structure.

Also under MCMC, note that ultimately we are troubled with the curse of dimensionality as we
will never be able to get an accurate approximation of the joint posterior distribution in a setting
with several hundred parameters.

As MCMC samples yield confidence bands for parameter estimates, they can easily be checked for significance at every desired
level, i.e.,~parameters are not significant if confidence bands cover the value zero. In our subsequent examples, almost all
parameters are significant. Given internal mortality data, these confidence bands for parameter estimates can also be used to
test whether parameters significantly deviate from officially published life tables. On the other hand, MCMC is perfectly
applicable to sparse data as life tables can be used as prior distributions with confidence bands providing an estimate for parameter uncertainty which increase with fewer data points.

\subsection{Estimation via Matching of Moments}\label{sec:MM}

Finally, we provide a \emph{matching of moments} approach which allows easier estimation of parameters
but which is less accurate.
Therefore, we suggest this approach solely to be used to obtain starting values for the other, more sophisticated
estimation procedures.
In addition, matching of moments approach needs simplifying assumptions to guarantee independent and
identical random variables over time.

\begin{assumption}[i.i.d.~setting]\label{constant_weight}
	Given Definitions\/ \ref{def:additiveModel}--\ref{simple_assumptions}, assume death counts\/
	$(N_{a,g,k}(t))_{t\in U}$ to be i.i.d.~with\/ $E_{a,g}:=E_{a,g}(1)=\dots=\mathbb{E}_{a,
	g}(T)$ and\/ $m_{a,g}:=m_{a,g}(1)=\dots=m_{a,
	g}(T)$, as well as\/ $w_{a,g,k}:=w_{a,g,k}(1)=\dots=w_{a,
	g,k}(T)$.
\end{assumption}

To achieve such an i.i.d.~setting, transform deaths $N_{a,g,k}(t)$
such that Poisson mixture intensities are constant over time via
\[
	N'_{a,g,k}(t):=\bigg\lfloor\frac{E_{a,g}(T)
	m_{a,g}(T) w_{a,g,k}(T)}{E_{a,g}(t)
	m_{a,g}(t) w_{a,g,k}(t)}
	N_{a,g,k}(t)\bigg\rfloor\,,\quad t\in U\,,
\]
and, correspondingly, define $E_{a,g}:=E_{a,g}(T)$,
as well as $m_{a,g}:=m_{a,g}(T)$ and $w_{a,g,k}:=w_{a,g,k}(T)$. Using~this modification,
	we manage to remove long term trends and keep $E_{a,g}(t)$,
	$m_{a,g}(t)$ and $w_{a,g,k}(t)$
	constant over time.
	
	Estimates $\hat{m}^{\mathrm{MM}}_{a,g}(t)$
	for central death
	rates $m_{a,g}(t)$ can be obtained via
	minimising mean squared error to crude death rates which,
	if parameters $\zeta$, $\eta$ and $\gamma$ are previously fixed, can be obtained by  regressing
	\[
		(F^{\mathrm{Lap}})^{-1}\bigg(\frac{\sum_{k=0}^K \widehat N'_{a,g,k}(t)}{E_{a,g}}\bigg)-\gamma_{a-t}
	\]
	on $\mathcal{T}_{\zeta_{a,g},\eta_{a,g}}(t)$.
	Estimates
	$\hat{u}^{\mathrm{MM}}_{a,g,k},\hat{v}^{\mathrm{MM}}_{a,g,k},\hat{\phi}^{\mathrm{MM}}_{k},\hat{\psi}^{\mathrm{MM}}_{k}$  for parameters $u_{a,g,k},v_{a,g,k},\phi_k,\psi_k$ via
	minimising the mean squared error to crude death rates which
	again, if parameters $\phi$ and $\psi$ are previously fixed, can~be obtained by regressing
	$\log(\widehat N'_{a,g,k}(t))-\log(E_{a,g} \hat{m}^{\mathrm{MM}}_{a,g}(t))$
	on $\mathcal{T}_{\phi_k,\psi_k}(t)$.
	Estimates
	$\hat{w}^{\mathrm{MM}}_{a,g,k}(t)$
	are then given by (\ref{eq:WeightFamily}).
	
	Then, define unbiased estimators for weights
	$W^*_{a,g,k}(t):={N'_{a,g,k}(t)}/{E_{a,g} m_{a,g}}$,
	as well as
	\[
		\overline{W}^{*}_{a,g,k}:=\frac{1}{T}\sum_{t=1}^T{W}^{*}_{a,g,k}(t)\,.
	\]

	In particular, we have
	$\mathbb{E}[\overline{W}^{*}_{a,g,k}]=\mathbb{E}[{W}^{*}_{a,g,k}(t)]= w_{a,g,k}$.
		
\begin{lemma}
	Given Assumptions \ref{simple_assumptions}, define
	\[
		\widehat{\Sigma}_{a,g,k}^2=\frac{1}{T-1}\sum_{t=1}^T\big({W}^{*}_{a,g,k}(t)-\overline{W}^{*}_{a,g,k}\big)^2\,,
	\]
	for all\/ $a\in \{0,\dots,A\}$, $g\in\{\mathrm{f},\mathrm{m}\}$ and\/ $k\in\{0,\dots,K\}$.
	Then,
	\begin{equation}\label{var_single}
		\mathbb{E}[{\widehat{\Sigma}_{a,g,k}^2}]=\var{{W}^{*}_{a,g,k}(1)}
		=\frac{w_{a,g,k}}{ E_{a,g} m_{a,g}}+\sigma^2_k w_{a,g,k}^2\,.
	\end{equation}
\end{lemma}

\begin{proof}
	Note that $({W}^{*}_{a,g,k}(t))_{t\in U}$ is assumed to be an i.i.d.~sequence. Thus,
	since $\widehat{\Sigma}_{a,g,k}$ is an unbiased estimator for the standard deviation of
	${W}^{*}_{a,g,k}(1)$ and $\overline{W}^{*}_{a,g,k}$, see~\cite[Example~11.2.6]{statistic},
	we immediately get
	\[
		\mathbb{E}[{\widehat{\Sigma}_{a,g,k}^2}]=\text{Var}({\overline{W}^*_{a,g,k}(1)})=\text{Var}\bigg(\frac{N'_{a,g,k}(1)}{E_{a,g} m_{a,g}}\bigg)\,.
	\]

	Using the law of total variance as in \cite[Lemma 3.48]{schmock}, as well as  Definition \ref{simple_model} gives  
	\[
			E_{a,g}^2 m_{a,g}^2 \mathbb{E}[{\widehat{\Sigma}_{a,g,k}^2}]
			=\mathbb{E}[{\text{Var}({N'_{a,g,k}(1)}|{ \Lambda_k}})]+\text{Var}({\mathbb{E}[{N'_{a,g,k}(1)}|{ \Lambda_k}}])\,.
	\]

	Since $\text{Var}({N'_{a,g,k}(1)}|{ \Lambda_k})=\mathbb{E}[{N'_{a,g,k}(1)}|{ \Lambda_k}]=E_{a,g}m_{a,g} w_{a,g,k} \Lambda_k$ a.s., the equation above
	 gives the~result.
\end{proof}

Having obtained Equation (\ref{var_single}), we may define the following matching of moments estimates for
risk factor variances.

\begin{definition}[Matching of moments estimates for risk factor variances]\label{MMest}
	Given Assumption \ref{simple_assumptions}, the~\emph{matching of moments estimate}
	for $\sigma_k$ for all $k\in\{1,\dots,K\}$ is defined as
	\[
		\big(\hat{\sigma}_{k}^{\mathrm{MM}}\big)^2:=
		{\max\Biggl\lbrace0,\frac{\sum_{a=0}^A\sum_{g\in\{\mathrm{f},\mathrm{m}\}}\Big(\hat{\sigma}_{a,g,k}^2
		-\frac{ w^{\mathrm{MM}}_{a,g,k}(T)}{E_{a,g} m^{\mathrm{MM}}_{a,g}(T)}\Big)}{\sum_{a=0}^A\sum_{g\in\{\mathrm{f},\mathrm{m}\}} (w^{\mathrm{MM}}_{a,g,k}(T))^2}\Biggr\rbrace}
		\,,
	\]
	where $\hat{\sigma}_{a,g,k}^2$ is the estimate corresponding to estimator $\widehat{\Sigma}_{a,g,k}^2$.
\end{definition}

\section{Applications}\label{applications_mortality}

\subsection{Prediction of Underlying Death Causes}\label{sec:real}

As an applied example for our proposed stochastic mortality model, as well as for some further applications,  we take annual death data from Australia for the period 1987 to 2011. We fit our model
 using the matching of moments approach, as well as the maximum-likelihood approach
with Markov chain Monte Carlo (MCMC).
Data source for historical Australian population, categorised by age and gender, is taken from the
Australian Bureau
of Statistics\footnote{{\href{http://www.abs.gov.au/AUSSTATS/abs@.nsf/DetailsPage/3101.0Jun\%202013?OpenDocument}{http://www.abs.gov.au/}},
 accessed on May 10, 2016.} and data for the number of deaths categorised by death cause and divided into eight age categories, i.e., 50--54~years,
55--59~years, 60--64 years, 65--69~years, 70--74 years, 75--79 years, 80--84 years and 85+ years, denoted by $a_1,\dots,a_8$, respectively,
for each gender
is taken from the AIHW\footnote{{\href{http://www.aihw.gov.au/deaths/aihw-deaths-data/\#nmd}{http://www.aihw.gov.au/deaths/aihw-deaths-data/\#nmd}},
 accessed on May 10, 2016.}. The provided death data
is divided into 19 different death causes---based on the ICD-9 or ICD-10 classification---where we identify the following ten
of them with common non-idiosyncratic risk factors: \textit{\lq certain infectious and parasitic diseases\rq,
\lq neoplasms\rq, \lq endocrine, nutritional and metabolic diseases\rq, \lq mental and behavioural disorders\rq,
\lq diseases of the nervous system\rq, \lq circulatory diseases\rq, \lq diseases of the respiratory system\rq,
\lq diseases of the digestive system\rq,
\lq external~causes of injury and poisoning\rq,
\lq diseases of the genitourinary system\rq}. We merge the remaining eight death causes to idiosyncratic risk as
their individual contributions to overall death counts are small for all categories.
Data handling needs some care as there was a change in classification of  death data in 1997 as explained at the website of the Australian Bureau of Statistics\footnote{{\href{http://www.abs.gov.au/ausstats/abs@.nsf/Products/3303.0~2007~Appendix~Comparability+of+statistics+over+time+\%28Appendix\%29?OpenDocument}{http://www.abs.gov.au/}}, accessed on May 10, 2016.}.
Australia introduced the tenth revision of the International Classification of Diseases (ICD-10, following ICD-9) in 1997, with a transition period from 1997 to 1998. Within this period,
comparability factors are given in Table \ref{tab:ComparabilityFactors}.
Thus, for the period 1987 to 1996, death counts have to be multiplied by corresponding
comparability factors.
\begin{table}[ht]
	\begin{center}\footnotesize{
		\caption[Example Australia: Comparability factors ICD-9 to ICD-10.]{Comparability factors for ICD-9 to ICD-10.}
		\setlength{\extrarowheight}{0.5pt}
		\label{tab:ComparabilityFactors}
		\begin{tabular}{r|r}
			\hline\rule[-6pt]{0pt}{18pt}
			death cause & factor \\ 
			\hline\rule[10pt]{0pt}{0pt}
			infectious& 	1.25 \\
			neoplasms& 1.00 \\
			endocrine& 1.01 \\
			mental & 	0.78 \\
			nervous& 1.20 \\
			circulatory& 1.00 \\
			respiratory& 0.91 \\
	    digestive& 1.05 \\
			genitourinary& 1.14 \\
			external& 1.06 \\\rule[-5pt]{0pt}{0pt}
			not elsewhere (idio.)& 	1.00
			 \\\hline
		\end{tabular}}
	\end{center}
\end{table}

To reduce the number of parameters which have to be estimated, cohort effects are not considered, i.e.,~$\gamma=0$, and
trend reduction parameters are fixed with values $\zeta=\phi=0$ and $\eta=\psi=\frac{1}{150}$. This~corresponds to slow trend
reduction over the data and forecasting period (no acceleration) which makes the setting similar to the Lee--Carter model.
Moreover, we choose the arbitrary normalisation $t_0=1987$.
Results for a more advanced modelling of trend reduction are shown later in Section \ref{sec:forecaDP}. Thus,~within the maximum-likelihood framework,  
we end up with $394$ parameters, with 362 to be optimised.
For~matching of moments we follow the approach given in Section \ref{sec:MM}. Risk factor variances are then estimated via Approximations~(\ref{MAPappr_lambda}) and
(\ref{MAPappr_beta}) of the maximum a posteriori approach as they
give more reliable results than matching of moments.

Based on 40,000 MCMC steps with burn-in period of 10,000 we are able to derive estimates of all parameters where
starting values are taken from matching of moments, as well as (\ref{MAPappr_lambda}) and
(\ref{MAPappr_beta}). Tuning~parameters are frequently re-evaluated in the burn-in period. The execution time of our algorithm is roughly seven hours on a standard computer in \lq R\rq. Running several parallel MCMC chains
reduces execution times to several minutes. However, note that a reduction in risk factors  (e.g.,~one or zero risk factors for mortality modelling) makes estimation much quicker.

As an illustration, Figure \ref{fig:MCMCchainReal} shows MCMC chains of the variance
of risk factor for external causes of injury and poisoning  $\sigma^2_{9}$, as well as of the parameter $\alpha_{2,f}$ for death probability intercept of females aged $55$ to $59$ years. We observe in Figure \ref{fig:MCMCchainReal}
that stationary distributions of MCMC chains for risk factor variances
are typically right skewed. This indicates risk which is associated with underestimating
variances due to limited observations of tail events.

\begin{figure}[ht]
	\centering
		\captionsetup{aboveskip=0.2\normalbaselineskip}
		\includegraphics[width=0.8\textwidth]{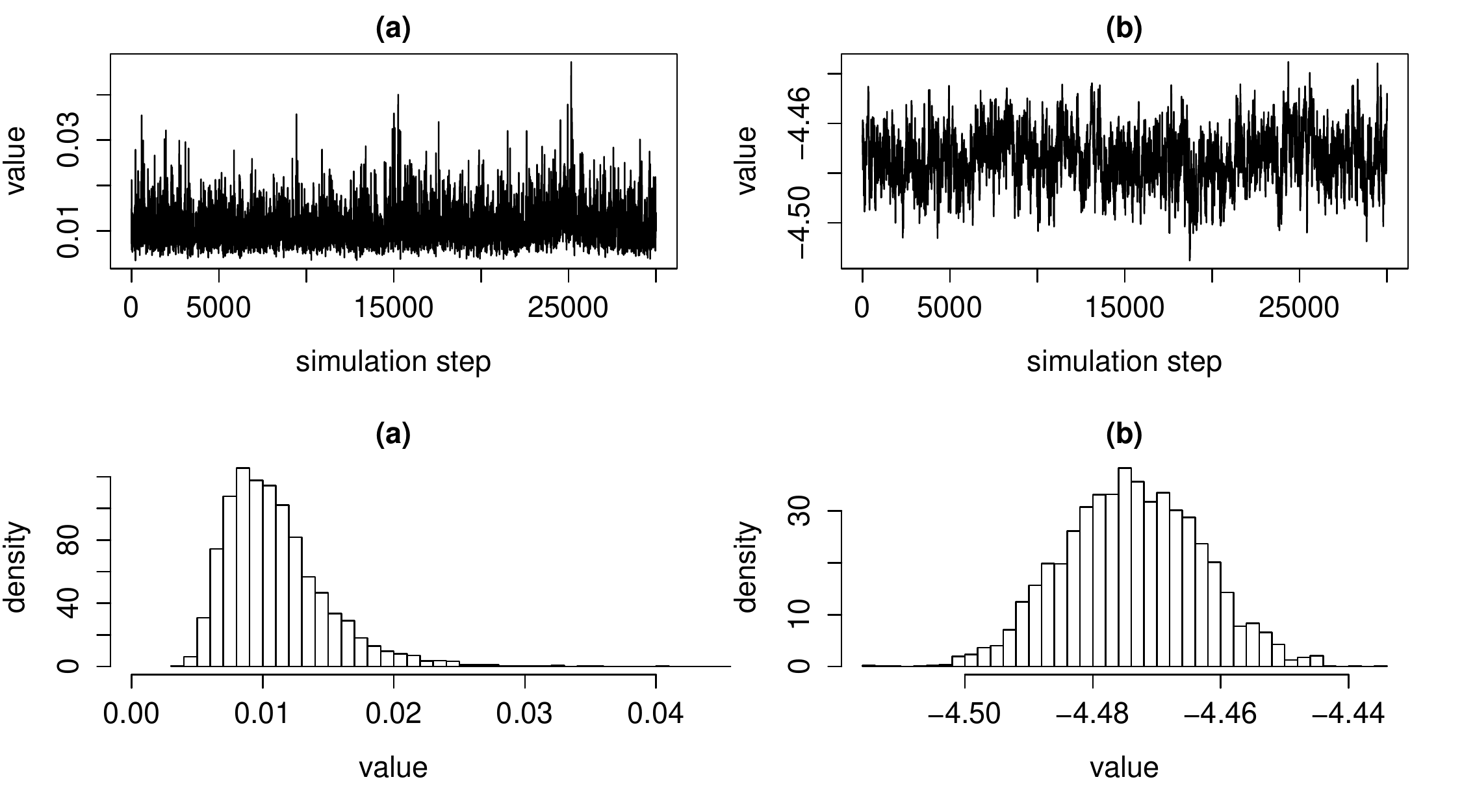}
		\caption[Example Australia: MCMC chains and density histograms for\/ $\sigma^2_{9}$ and\/ $\alpha_{2,\mathrm{f}}$.]{MCMC chains and corresponding density histograms for the variance of risk factor for deaths due to external causes of injury and poisoning\/
		$\sigma^2_{9}$ in subfigure \textbf{(a)} and for the death probability intercept parameter of females aged $55$ to $59$ years\/ $\alpha_{2,\mathrm{f}}$ in subfigure \textbf{(b)}.}
		\label{fig:MCMCchainReal}
	
\end{figure}
Table \ref{table_MCMCreal} shows estimates for risk factor standard deviations using matching of moments, Approximation (\ref{MAPappr_beta}), as well as mean estimates of MCMC
with corresponding 5\% and 95\% quantiles, as well as standard errors. First, Table \ref{table_MCMCreal} illustrates
that (\ref{MAPappr_lambda}) and (\ref{MAPappr_beta}), as well as matching of moments
estimates for risk factor standard deviations $\sigma$
are close to mean MCMC estimates. Risk factor standard deviations are small but tend to be higher
for death causes with just few deaths as statistical fluctuations in the data are higher compared to more frequent
death causes.
Solely estimates for the risk factor standard deviation of mental and behavioural disorders
give higher values. Standard errors, as defined in \cite[section 2.12.2]{pavel} with block size $50$,  for corresponding risk factor variances are consistently
less than 3\%.
We can use
the approximation given in (\ref{lambda_MAP}) to derive risk factor estimates over previous years.
For example, we observe increased risk factor realisations of diseases of the respiratory system over the years 2002 to 2004. This is mainly driven by many deaths due to influenza and pneumonia during that period.

\begin{table}[ht]
	\begin{center}\footnotesize{
		\caption[Example Australia: Estimates for risk factor standard deviations.]{Estimates for risk factor standard deviations\/ $\boldsymbol{\sigma}$ using matching of moments (MM), Approximation (\ref{MAPappr_beta}) (appr.) and MCMC mean estimates  (mean), as well as corresponding standard deviations (stdev.) and
		five and 95 percent quantiles (5\%~and 95\%).}
		\setlength{\extrarowheight}{0.5pt}
		\label{table_MCMCreal}
		\begin{tabular}{r|p{1cm}p{1cm}p{1cm}p{1cm}p{1cm}p{1cm}}
			\cline{2-7}\rule[-6pt]{0pt}{18pt}
			&MM&appr.&mean&5\%&95\%&stdev. \\ 
			\hline\rule[10pt]{0pt}{0pt}
			infectious&0.1932 & 0.0787 &	0.0812		&0.0583	&0.1063	&0.0147	\\
			neoplasms&0.0198 & 0.0148	&0.0173		&0.0100	&0.0200 &0.0029	\\
			endocrine&0.0743 & 0.0340		&0.0346 &0.0245	&0.0469 &	0.0068		\\
			mental&0.1502 & 0.1357		&0.1591 &	0.1200 &	0.2052	&	0.0265\\
			nervous &0.0756 & 0.0505		&0.0557 &	0.0412	&0.0728 &	0.0098	\\
			circulatory&0.0377 & 0.0243	&	0.0300	&	0.0224	&0.0387	 &0.0053\\
			respiratory&0.0712 & 0.0612	&	0.0670 &	0.0510	&0.0866	&	0.0110\\
			digestive&0.0921 & 0.0645		&0.0728		&0.0548	&0.0943	&0.0123\\
			external&0.1044 & 0.0912	&	0.1049 &	0.0787 &	0.1353 &	0.0176
\\\rule[-5pt]{0pt}{0pt}
			genitourinary&0.0535 & 0.0284		&0.0245		&0.0141 &	0.0346 &0.0066
\\\hline
		\end{tabular}}
	\end{center}
\end{table}

Assumption (\ref{eq:WeightFamily}) provides a joint forecast of all death cause intensities, i.e.,~weights, simultaneously---in contrast to standard procedures where projections are made for each death cause separately.
Throughout the past decades we have observed drastic shifts in crude death rates due to certain death causes over the past decades.   This fact can be be illustrated by our model as shown in
Table \ref{tab:deathCauses}. This table lists weights $w_{\textrm{a},\textrm{g},k}(t)$ for all death causes estimated for 2011, as well as forecasted for 2031 using (\ref{eq:WeightFamily}) with MCMC mean estimates for males and females aged between $80$ to $84$ years.
Model forecasts suggest that if these trends in weight changes persist, then~the future gives a whole new picture of mortality. First, deaths due to circulatory diseases are expected to decrease whilst neoplasms will become
the leading death cause over most age categories.
Moreover, deaths due to mental and behavioural disorders are expected to rise considerably for older~ages.
 High~uncertainty in forecasted weights
is reflected by wide confidence intervals (values in brackets) for the risk factor of mental and behavioural disorders. These confidence intervals are derived from corresponding MCMC chains
and, therefore, solely reflect uncertainty associated with parameter estimation. Note that results for estimated trends depend on the length of the data period as short-term trends might not coincide with mid- to long-term trends.
Further results can be found in~\cite{forecaDeathCauses2015}.

\begin{table}[ht]
	\begin{center}\footnotesize{
		\caption[Example Australia: Estimated and forecasted weights.]{Estimated weights for all death causes in years\/  2011, 2021 and\/ 
2031 using (\ref{eq:WeightFamily}) with MCMC mean estimates for ages\/ 60 to\/ 64 years (left) and\/ 80 to\/ 84 years (right) for both genders. Five and\/ 95 percent quantiles for the year\/ 2031 are given in brackets.}
		\setlength{\extrarowheight}{0.5pt}
		\label{tab:deathCauses}
		\begin{tabular}{r|rrr|rrr}
		
			\cline{2-7}\rule[-6pt]{0pt}{18pt}
			&\multicolumn{3}{ c | }{60 to 64 years}&\multicolumn{3}{ c  }{80 to 84 years}\\
			\cline{2-7}\rule[-6pt]{0pt}{18pt}
			  & 2011 & 2021 & 2031 \scriptsize{(quant.)} & 2011 & 2021 & 2031 \scriptsize{(quant.)} \\
			\cline{2-7}\rule[-6pt]{0pt}{18pt}
		& \multicolumn{6}{ c }{male}\\\hline\rule[12pt]{0pt}{0pt}
		neoplasms & 0.499& 0.531 &0.547 $\big(\substack{0.561\\ 0.531}\big)$& 0.324 &0.359& 0.378 $\big(\substack{0.392\\ 0.364}\big)$\\[0.8ex]
		circulatory & 0.228 &0.165 &0.116 $\big(\substack{0.123 \\0.109}\big)$& 0.325 &0.242 &0.173 $\big(\substack{0.181 \\0.164}\big)$\\[0.8ex]
		external& 0.056& 0.060& 0.062 $\big(\substack{0.073 \\0.053}\big)$& 0.026& 0.028 &0.028 $\big(\substack{0.033\\ 0.024}\big)$\\[0.8ex]
		respiratory& 0.051 &0.043& 0.036 $\big(\substack{0.040 \\0.032}\big)$& 0.106 &0.101 &0.092 $\big(\substack{0.101\\ 0.083}\big)$\\[0.8ex]
		endocrine&0.044 &0.053 &0.062 $\big(\substack{ 0.070 \\0.055}\big)$& 0.047& 0.062& 0.077 $\big(\substack{0.084\\ 0.070}\big)$\\[0.8ex]
		digestive &0.041& 0.039 &0.036 $\big(\substack{0.040\\ 0.031}\big)$&  0.027 &0.024 &0.020 $\big(\substack{0.023\\ 0.018}\big)$\\[0.8ex]
		nervous & 0.029 &0.040 &0.052 $\big(\substack{0.061\\ 0.045}\big)$& 0.045& 0.054 &0.061 $\big(\substack{0.068\\ 0.055}\big)$\\	[0.8ex]	
		not elsewhere (idio.) &0.018 &0.023& 0.028 $\big(\substack{0.034\\ 0.023}\big)$&0.015& 0.017 &0.018 $\big(\substack{0.020\\ 0.016}\big)$\\[0.8ex]
		infectious & 0.014 &0.019 &0.025 $\big(\substack{0.033 \\0.020}\big)$& 0.015 &0.019 &0.022 $\big(\substack{0.027\\ 0.019}\big)$\\[0.8ex]
		mental & 0.013 &0.019 &0.027 $\big(\substack{0.036 \\0.019}\big)$&  0.041 &0.068&0.105 $\big(\substack{0.130\\ 0.078}\big)$\\[0.8ex]\rule[-7pt]{0pt}{0pt}
		genitourinary & 0.008& 0.008& 0.008 $\big(\substack{0.010\\ 0.006}\big)$&  0.028& 0.027 &0.025 $\big(\substack{0.028 \\0.023}\big)$\\[0.8ex]
		\hline\rule[-6pt]{0pt}{18pt}
		& \multicolumn{6}{ c }{female}\\\hline\rule[12pt]{0pt}{0pt}
		neoplasms & 0.592& 0.628& 0.648 $\big(\substack{0.662 \\0.629}\big)$&0.263 &0.293 &0.303 $\big(\substack{0.319\\ 0.288}\big)$\\[0.8ex]
		circulatory & 0.140 &0.092& 0.060 $\big(\substack{0.065 \\0.055}\big)$&  0.342 &0.233 &0.149 $\big(\substack{0.158\\ 0.140}\big)$\\[0.8ex]
		respiratory & 0.072& 0.071 &0.069   $\big(\substack{0.078 \\0.060}\big)$&  0.100& 0.116& 0.126 $\big(\substack{0.139\\ 0.113}\big)$\\[0.8ex]
		endocrine& 0.038 &0.038 &0.037 $\big(\substack{0.043\\ 0.032}\big)$& 0.051& 0.061 &0.068 $\big(\substack{0.074\\ 0.061}\big)$\\[0.8ex]
		nervous & 0.036 &0.043& 0.051 $\big(\substack{0.060 \\0.043}\big)$&0.054& 0.068 &0.080 $\big(\substack{0.089\\ 0.071}\big)$\\[0.8ex]
		external  & 0.035 &0.033 &0.032 $\big(\substack{0.038\\ 0.026}\big)$& 0.024 &0.025& 0.023 $\big(\substack{0.027 \\0.020}\big)$\\[0.8ex]
		digestive &0.031& 0.028& 0.024 $\big(\substack{ 0.029\\ 0.020}\big)$& 0.034& 0.029 &0.023 $\big(\substack{0.027\\ 0.020}\big)$\\[0.8ex]
		not elsewhere (idio.)& 0.022 &0.023& 0.023 $\big(\substack{0.028\\ 0.019}\big)$& 0.023 &0.025 &0.024 $\big(\substack{0.027 \\0.022}\big)$\\[0.8ex]
		infectious &  0.014& 0.017 &0.020  $\big(\substack{0.027 \\0.015}\big)$&  0.017& 0.021& 0.024 $\big(\substack{0.028\\ 0.020}\big)$\\[0.8ex]
		mental  & 0.012& 0.019 &0.032 $\big(\substack{0.046\\ 0.021}\big)$& 0.062& 0.102 &0.155 $\big(\substack{0.188 \\0.118}\big)$\\[0.8ex]\rule[-7pt]{0pt}{0pt}
		genitourinary &  0.009 &0.007 &0.005 $\big(\substack{0.006 \\0.004}\big)$&0.029& 0.028 &0.026 $\big(\substack{0.028\\ 0.023}\big)$\\[0.8ex]
\hline
		\end{tabular}}
	\end{center}
\end{table}

\subsection{Forecasting Death Probabilities}\label{sec:forecaDP}

Forecasting death probabilities and central death rates within our proposed model is straight forward using
(\ref{eq:PDFamily}). In the special case with just idiosyncratic risk, i.e., $K=0$, death indicators can be assumed to be Bernoulli distributed instead of being Poisson distributed in which case we may write the likelihood function in the form
\[
\begin{split}
		\ell^{\mathrm{B}}(&\widehat{N}|\alpha,\beta,\zeta,\eta,\gamma)= \\ &
		\prod_{t=1}^T \prod_{a=0}^A\prod_{g\in\{\mathrm{f},\mathrm{m}\}} \binom {E_{a,g}(t)} {\widehat{N}_{a,g,0}(t)}
		m_{a,g}(t)^{\widehat{N}_{a,g,0}(t)} (1-m_{a,g}(t))^{E_{a,g}(t)-\widehat{N}_{a,g,0}(t)}\,,
\end{split}
\]
with $0\leq \widehat{N}_{a,g,0}(t)\leq E_{a,g}(t)$.
Due to possible overfitting, derived estimates may not be sufficiently smooth across age categories $a\in\{0,\dots,A\}$. Therefore, if we switch to a Bayesian setting, we may use regularisation via prior distributions to obtain stabler results. To guarantee smooth results and a~sufficient stochastic foundation, we suggest the usage of Gaussian priors with mean zero and a specific correlation structure,
i.e., $\pi(\alpha,\beta,\zeta,\eta,\gamma)=\pi(\alpha)\pi(\beta)\pi(\zeta)\pi(\eta)\pi(\gamma)$ with
	\begin{equation}\label{prior_TV}
		\log\pi(\alpha):=	-c_\alpha \sum_{g\in\{\mathrm{f},\mathrm{m}\}}\bigg(
			\sum_{a=0}^{A-1}  (\alpha_{a,g}-\alpha_{a+1,g})^2
			+\varepsilon_\alpha \sum_{a=0}^{A} \alpha_{a,g}^2\bigg)+\log(d_\alpha)\,,
	\end{equation}
	$c_\alpha,d_\alpha,\varepsilon_\alpha>0$, and correspondingly for $\beta$, $\zeta$, $\eta$ and $\gamma$. Parameters $c_\alpha$ (correspondingly
	for $\beta$, $\zeta$, $\eta$ and $\gamma$) is a scaling parameters and directly associated with the variance of Gaussian priors while
	normalisation-parameter $d_\alpha$ guarantees that $\pi(\alpha)$ is a proper Gaussian density.
	Penalty-parameter $\varepsilon_\alpha$ scales the correlation amongst neighbour parameters in the sense that the lower it gets, the higher the correlation. The more we increase $c_\alpha$
	the stronger the influence of,  or the believe in  the prior distribution. This particular prior density penalises
 deviations from the ordinate which is
	a mild conceptual shortcoming as this does not accurately reflect our prior believes. Setting  $\varepsilon_\alpha=0$ gives an improper
	prior with uniformly distributed (on $\mathbb{R}$) marginals such that we gain that there is no prior believe in expectations of parameters but, simultaneously, lose the presence of  variance-covariance-matrices and asymptotically get perfect positive correlation across parameters of different ages. Still, whilst lacking theoretical properties, better fits to data are obtained by setting $\varepsilon_\alpha=0$.
	For example, setting $\varepsilon_\alpha=\varepsilon_\beta=10^{-2}$ and
	$\varepsilon_\zeta=\varepsilon_\eta=\varepsilon_\gamma=10^{-4}$
	yields a  prior correlation structure which decreases with higher age differences and which is always positive as given in
	subfigure \textbf{(a)} of Figure \ref{fig:corr}.
	 \begin{figure}[ht]
	\begin{center}
		\captionsetup{aboveskip=0.2\normalbaselineskip}
		\includegraphics[width=0.88\textwidth]{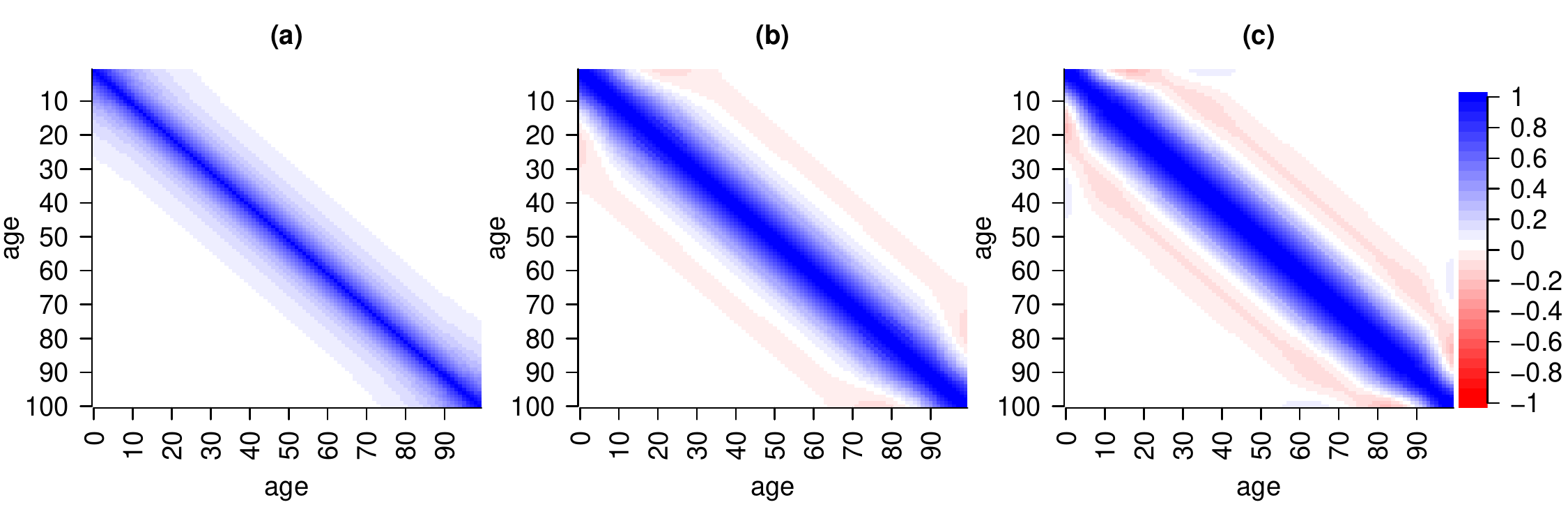}
		\caption[Correlation structure for Gaussian priors.]{Correlation structure of Gaussian priors with penalisation for deviation from ordinate with $\varepsilon = 1/100$ in subfigure \textbf{(a)}, straight line with $\varepsilon = 1/2000$ in subfigure \textbf{(b)}, and parabola $\varepsilon = 1/50000$ in subfigure \textbf{(c)}.}
		\label{fig:corr}
	\end{center}
\end{figure}
	  There exist many other reasonable choices for Gaussian prior densities. For example, replacing graduation terms
	  $(\alpha_{a,g}-\alpha_{a+1,g})^2$ in (\ref{prior_TV}) by higher order differences
		of the form $\big(\sum_{\nu=0}^k(-1)^\nu {k\choose\nu} \alpha_{a,g+\nu}\big)^2$
	 yields a penalisation for deviations from a straight line with $k=2$, see subfigure \textbf{(b)} in Figure \ref{fig:corr}, or from a parabola with $k=3$, see subfigure \textbf{(c)} in Figure \ref{fig:corr}.
	  The usage of higher order differences for graduation of statistical estimates goes back to
		the Whittaker--Henderson method. Taking $k=2,3$ unfortunately yields negative prior correlations
	  amongst certain parameters which is why we do not recommend their use. Of course, there exist many further possible
	  choices for prior distributions.
	  However, in our example, we set $\varepsilon_\alpha=\varepsilon_\beta=\varepsilon_\zeta=\varepsilon_\eta=\varepsilon_\gamma=0$ as this yields accurate results whilst still being reasonably smooth.
	
	 An optimal choice of regularisation parameters $c_\alpha,c_\beta,c_\zeta,c_\eta$ and $c_{\gamma}$ can be obtained by cross-validation.

Results for Australian data
from 1971 to 2013 with $t_0=2013$ are given in Figure \ref{fig:PDforecast3}. Using~MCMC we
derive estimates for logarithmic central death rates $\log m_{a,g}(t)$
with corresponding forecasts, mortality trends  $\beta_{a,g}$, as well as trend reduction parameters $\zeta_{a,g},\eta_{a,g}$ and cohort effects $\gamma_{a-t}$. As~we do not assume common stochastic risk factors,
the MCMC algorithm we use can be implemented very efficiently such that $40\,000$ samples from the posterior
distribution of all parameters are derived within a minute.
We observe negligible parameter uncertainty due to a long period of data. Further,
regularisation parameters obtained by cross-validation are given by  $c_\alpha=500$, $c_\beta=c_\eta=30,000 c_\alpha$, $c_\zeta=c_\alpha/20$ and $c_\gamma=1000 c_\alpha$.
\begin{figure}[ht]
	\centering
		
		\includegraphics[width=0.85\textwidth]{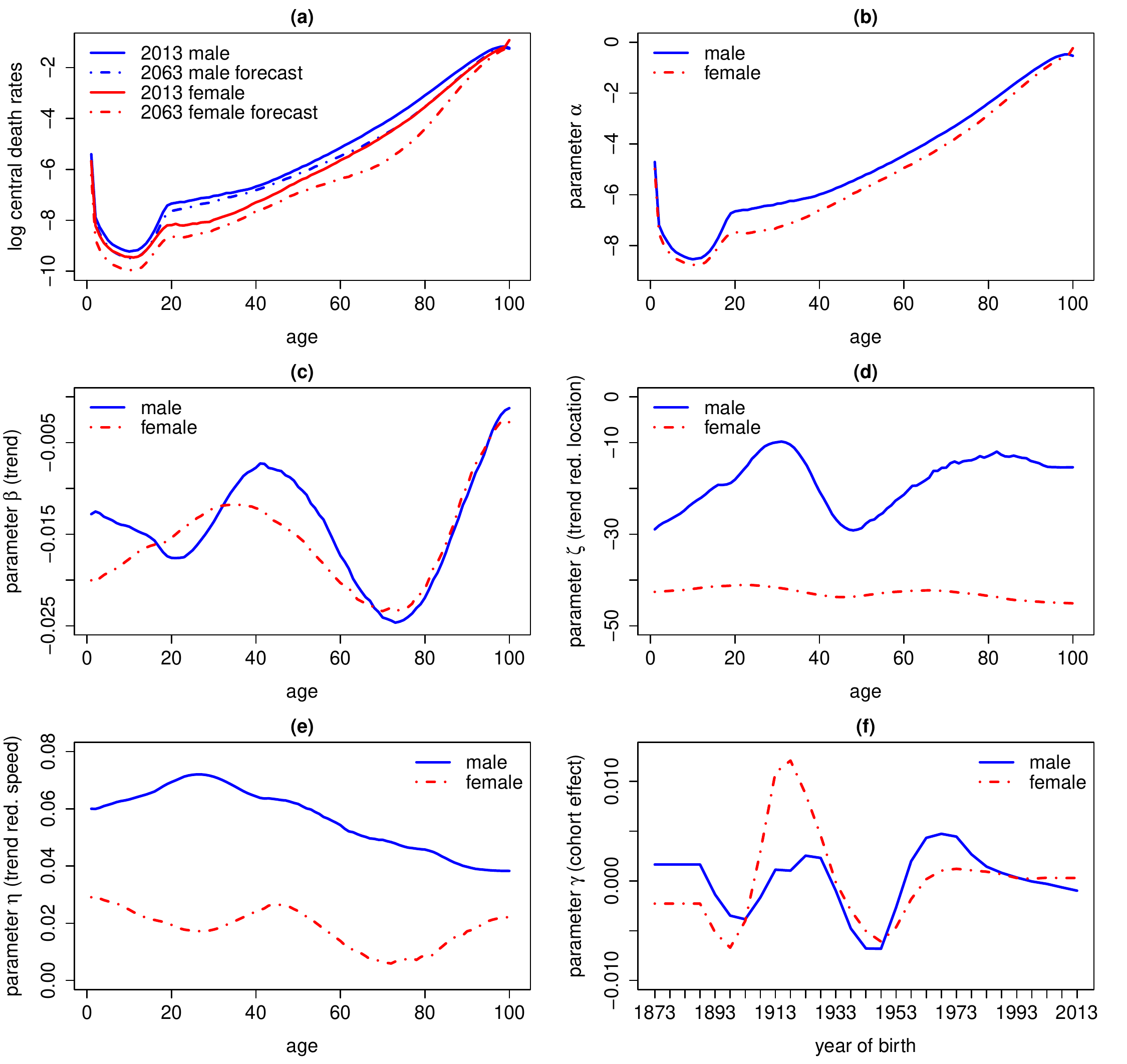}
		\caption{Logarithm of death central death rates (a) for 2013 and forecasts for 2063 in Australia as well as parameter values for $\alpha,\beta,\zeta,\eta$ and $\gamma$ in subfigures \textbf{(b)}, \textbf{(c)}, \textbf{(d)}, \textbf{(e)} and \textbf{(f)}, respectively.}
		\label{fig:PDforecast3}
	
\end{figure}
We can draw some immediate conclusions. Firstly, we see
an overall improvement in mortality over all ages where the trend is particularly strong
for young ages and ages between 60  and 80 whereas the trend vanishes towards the age of 100, maybe implying a natural barrier for life expectancy. Due~to sparse data the latter conclusion should be treated with the utmost caution.
 Furthermore, we see the classical
hump of increased mortality driven by accidents around the age of 20 which is more
developed for males.

Secondly, estimates for $\zeta_{a,g}$ suggest that trend acceleration switched to trend reduction
throughout the past 10 to 30 years for males while for females this transition already took place 45 years ago.
However, note that parameter uncertainty (under MCMC) associated with  $\zeta_{a,g}$ is high, particularly if estimates are not regularised.
 Estimates for
$\eta_{a,g}$ show that the speed of trend reduction is much stronger for males than for females.
Estimates for $\gamma_{a-t}$ show that the cohort effect is particularly strong (in the sense of increased mortality) for the generation born between 1915 and 1930 (probably
associated with World War II) and particularly weak for the generation born around 1945. However, considering cohort effects makes estimation and forecasts significantly less stable for the used data, which is why we recommend to set $\gamma_{a-t}=0$.

Based on forecasts for death probabilities, expected future life time can be estimated. To be consistent concerning longevity risk, mortality trends have to
be included as a 60-year-old today will probably not
have as good medication as a 60-year-old in several decades. However, it seems that this is not
the standard approach in the literature.
Based on the definitions above, expected (curtate) future life time of a person at date $T$ is given by
$e_{a,g}(T)=\mathbb{E}[{K_{a,g}(T)}]=\sum_{k=1}^\infty
	{}_k p_{a,g}(T)$,
where survival probabilities over $k\in\mathbb{N}$ years are given by
${}_k p_{a,g}(T):=\prod_{j=0}^{k-1} \big(1-q_{a+j,g}(T+j)\big)$
and where $K_{a,g}(T)$ denotes the number of completed future years lived by
a person of particular age and gender at time~$T$. Approximating death probabilities by central death rates, for newborns in Australia we get a life expectancy of roughly $83$ years for males and $89.5$ for females born in 2013, see Table \ref{table_futureLifeTime}. Thus,
comparing these numbers to a press release from October 2014 from the Australian Bureau of Statistics\footnote{{\href{http://www.abs.gov.au/ausstats/abs@.nsf/mediareleasesbyReleaseDate/51FD51C3FC56234DCA257EFA001AE940?OpenDocument}{http://www.abs.gov.au/ausstats/abs@.nsf/mediareleasesbyReleaseDate/51FD51C3FC56234DCA257EFA001AE940?OpenDocument}}, accessed on May 10, 2016.} saying that
\lq Aussie men now expected to live past 80\rq~and \lq improvements in expected lifespan for women has since slowed down, increasing by around four~years over the period---it's 84.3 now\rq, our results show a
much higher life expectancy due to the consideration of
mortality trends.

\begin{table}[ht]
\centering \footnotesize{
		\caption{Curtate future life time $e_{a,g}(T)$ for males and females in 2013.}
		\setlength{\extrarowheight}{0.5pt}
		\label{table_futureLifeTime}
		\begin{tabular}{r | rrrrr}
			\cline{2-6}\rule[-6pt]{0pt}{18pt}
			{age in} $2013$ & $0$ (Newborn) & $20$ & $40$ & $60$ & $80$\\
			\hline\rule[10pt]{0pt}{0pt}
			male & $83.07$ & $63.33$ & $43.62$ & $24.44$ & $8.26$\\\rule[-5pt]{0pt}{0pt}
			female & $89.45$ & $69.05$ & $48.20$ & $27.76$ & $9.88$\\	\hline
		\end{tabular}}	
\end{table}	
	
\section{A Link to the Extended CreditRisk$^+$ Model and Applications}\label{modelling}

\subsection{The ECRP Model} In this section we establish the connection from our proposed stochastic mortality model to the risk aggregation model extended CreditRisk$^+$ (abbreviated as ECRP), as given in \cite[section 6]{schmock}.
	
	\begin{definition}[Policyholders and number of deaths]\label{def2}
	Let $\{1,\dots,E\}$ with $E\in\mathbb{N}$ denote
	the set of \emph{people} (termed as policyholders in light of insurance applications) in the portfolio
	and let random variables $N_1,\dots,N_E:\Omega\rightarrow\mathbb{N}_{0}$
	indicate the \emph{number of deaths} of each policyholder in the following period.
	The event $\{N_i=0\}$ indicates survival of person $i$ whilst $\{N_i\geq1\}$ indicates death.
	\end{definition}

	\begin{definition}[Portfolio quantities]\label{def1}
	\textls[-15]{Given Definition \ref{def2}, the independent random vectors
	${Y}_1,\dots,{Y}_E:\Omega\rightarrow\mathbb{N}_0^d$ with $d\geq 1$ dimensions
	denote \emph{portfolio quantities}
	 within the following period given deaths of policyholders,}
	i.e., on $\{N_i\geq 1\}$ for all $i\in\{1,\dots,E\}$, and
	are independent of ${N}_1,\dots,{N}_E$.
	\end{definition}
	
	\begin{remark}\label{Rem:model}(Portfolio quantities).
		\begin{enumerate}

			\item[(a)] For applications in the context of internal models we may set $Y_i$ as the best
			estimate liability, i.e.,~discounted~future cash flows, of policyholder $i$ at the end of the period.
			Thus, when using stochastic discount factors or contracts with optionality, for example,
			portfolio quantities may be stochastic.
			
			\item[(b)] In the context of portfolio payment analysis we may set $Y_i$ as the payments
			(such as annuities) to $i$ over the next period.
			We may include premiums in a second dimension in order to get joint distributions of premiums and payments.
			\item[(c)] For applications in the context of mortality estimation and projection we set $Y_i = 1$.
				
		\item[(d)] Using discretisation which preserves expectations (termed as stochastic rounding in \cite[section~6.2.2]{schmock},
				we may assume ${Y}_i$ to be $[ 0,\infty)^d$-valued .
		\end{enumerate}
	\end{remark}
	
	\begin{definition}[Aggregated portfolio quantities]\label{defS}
		Given Definitions \ref{def2} and \ref{def1}, aggregated portfolio quantities due to deaths are given by
		\[
      			S:=\sum_{i=1}^{E}\sum_{j=1}^{N_{i}}{Y}_{i,j}\,,
		\]
		where $({Y}_{i,j})_{j\in \mathbb{N}}$ for every $i\in\{1,\dots,E\}$
		is an i.i.d.~sequence of random variables with the same distributions as~${Y}_{i}$.
	\end{definition}

	\begin{remark}\label{rem:TotalLoss} In the context of term life insurance contracts, for example,
			${S}$ is the sum of best estimates of payments and premiums which
			are paid and received, respectively, due to deaths of policyholders, see Section \ref{sec:SII}. In the context of annuities,
			${S}$ is the sum of best estimates of payments and premiums which
			need not be paid and are received, respectively, due to deaths of policyholders. Then, small values of ${S}$, i.e., the left tail of its distribution,
			is the part of major interest and major risk.
	\end{remark}

	It is a demanding question how to choose
	the modelling setup such that the
	distribution of $S$ can be derived efficiently and accurately. Assuming $N_i$ to be Bernoulli distributed
	is not suitable for our intended applications as computational complexity explodes.
	Therefore, to make the modelling setup applicable in practical situations and to ensure a flexible handling in terms of
	multi-level dependence, we introduce the ECRP model which is
	based on  extended CreditRisk$^+$, see  \cite[section 6]{schmock}.
	
	\begin{definition}[The ECRP model]\label{simple_model} Given Definitions \ref{def2} and \ref{def1},
		the ECRP model satisfies the following additional assumptions:
		\begin{enumerate}
			 \item[(a)] \label{rf_property} Consider independent random common risk factors $\Lambda_{1},\dots,\Lambda_{K}:\Omega\rightarrow[ 0,\infty)$ which have a gamma distribution with mean $e_k=1$ and
				variance $\sigma_k^2> 0$,  i.e., with shape and
		         	inverse scale parameter~$\sigma^{-2}_k$. Also the degenerate case with
				$\sigma^2_k= 0$ for $k\in\{1,\dots,K\}$ is allowed.
				Corresponding weights \mbox{$w_{i,0},\dots,w_{i,K}\in[ 0,1]$} for every policyholder $i\in\{1,\dots,E\}$. Risk index zero represents idiosyncratic risk and we require
	$w_{i,0}+\dots +w_{i,K}=1$.
	
       			\item[(b)]\label{n0_property} Deaths $N_{1,0},\dots,N_{E,0}:\Omega\rightarrow\mathbb{N}_0$ are independent
			from one another, as well as all other
				random variables and, for all $i\in\{1,\dots,E\}$, they are Poisson distributed with
				intensity $m_i w_{i,0}$, i.e.,
				\[
					\mathbb{P}\bigg(\bigcap_{i=1}^E \{N_{i,0}=\widehat{N}_{i,0}\}\bigg)=\prod_{i=1}^E e^{-m_i w_{i,0}}\frac{(m_i w_{i,0})^{\widehat{N}_{i,0}}}{\widehat{N}_{i,0}!}\,,\quad \widehat{N}_{1,0},\dots,\widehat{N}_{E,0}\in\mathbb{N}_0\,.
				\]
		         \item[(c)]\label{conditional_property} Given  risk factors, deaths $(N_{i,k})_{i\in\{1,\dots,E\},k\in\{1,\dots,K\}}:\Omega\rightarrow\mathbb{N}_0^{E\times K}$ are independent and,
			 	for every policyholder $i\in\{1,\dots,E\}$ and $k\in\{1,\dots,K\}$, they
				are Poisson distributed with random intensity $m_{i} w_{i,k} \Lambda_{k}$, i.e.,
				\[
					\mathbb{P}\bigg(\bigcap_{i=1}^E\bigcap_{k=1}^K \{N_{i,k}=\widehat{N}_{i,k}\}\,
					\bigg|\,\Lambda_1,\dots,\Lambda_K\bigg)=
					\prod_{i=1}^E\prod_{k=1}^K e^{-m_{i} w_{i,k} \Lambda_{k}}
					\frac{(m_{i} w_{i,k} \Lambda_{k})^{\widehat{N}_{i,k}}}{\widehat{N}_{i,k}!}\quad\textrm{a.s.,}
		                	\]
				for all $n_{i,k}\in\mathbb{N}_0$.
			\item[(d)]\label{sum_property}  For every policyholder $i\in\{1,\dots,E\}$,
			the total number of deaths $N_i$ is split up additively according
			to risk factors as
			$N_i=N_{i,0}+\dots+N_{i,K}$.
			Thus, by model construction, $\mathbb{E}[{N_i}]=m_i (w_{i,0}+\dots+w_{i,K})=m_i$.
		\end{enumerate}
	\end{definition}
	
	Given Definition \ref{def2}, central death rates are given by $m_i=\mathbb{E}[{N_i}]$ and death probabilities, under~piecewise constant death rates, are given by $q_i=1-\exp(-m_i)$.
	
	\begin{remark}
		Assuming that central death rates and weights are equal
		for all policyholders for the same age and gender, it is
		obvious that the ECRP corresponds one-to-one to our proposed stochastic mortality model, as given in Definition \ref{def:additiveModel}, if risk factors are independent gamma distributed.
	\end{remark}
	
	In reality, number of deaths are Bernoulli random variables as each person can just die once.
	Unfortunately in practice, such an approach is
	not tractable for calculating P\&L distributions of large portfolios as execution times
	explode if numerical errors should be small.
	Instead, we will assume the number of deaths of each
	policyholder to be compound Poisson distributed.
	However, for estimation of life tables we will assume the number of deaths to be Bernoulli distributed. 	
	Poisson distributed deaths give an efficient way for calculating P\&L distributions
	using an algorithm based on Panjer's recursion, also for
	large portfolios, see \cite[section 6.7]{schmock}.
	The algorithm is basically due to~\cite{Giese} for which~\cite{Haaf}
	proved numerical stability. The relation to Panjer's recursion was first pointed out in
	~\cite[section 5.5]{panjer}.
	 \cite{schmock} in section 5.1 generalised the algorithm to the multivariate case with dependent risk factors
	and risk groups, based on the multivariate extension of Panjer's algorithm
	given by \cite{sundt}.
	The algorithm is numerically stable since just positive terms are added up.
	To avoid long execution times for implementations of extended CreditRisk$^+$ with large annuity
	portfolios, greater loss units and stochastic rounding, \mbox{see \cite[section 6.2.2]{schmock}}, can~be~used.
	
	However, the proposed model allows for multiple (Poisson) deaths of each policyholder and thus approximates the \lq real world\rq~with single (Bernoulli) deaths.
	From a theoretical point of view, this~is
	justified by the Poisson approximation and generalisations of it, see for example \cite{poisson_approx}.
	Since annual death probabilities for ages up to 85 are less than 10\%,
	multiple deaths are relatively unlikely for all major ages.
	However, implementations of
			this algorithm are significantly faster than Monte Carlo approximations for comparable error (Poisson against Bernoulli) levels.

			As an illustration we take a portfolio with $E=10,000$ policyholders having central death rate $m:= m_i = 0.05$ and payments $Y_i=1$.
			We then derive the distribution of $S$ using the ECRP model for the case with just idiosyncratic risk, i.e., $w_{i,0}=1$ and Poisson distributed deaths, and for the case with just one common stochastic risk factor $\Lambda_1$ with variance $\sigma_1=0.1$ and no idiosyncratic risk, i.e., $w_{i,1}=1$ with mixed Poisson distributed deaths.
			Then, using $50,000$ simulations of the corresponding model where $N_i$ is
			Bernoulli distributed or mixed Bernoulli distributed given truncated risk factor $\Lambda_1| \Lambda_1\leq \frac{1}{m}$, we compare the results of the ECRP model to Monte Carlo, respectively. Truncation of risk factors in the Bernoulli model is necessary as otherwise death probabilities may exceed one.
			We observe that the ECRP model drastically reduces execution times in \lq R\rq~at comparable error levels and leads to a~speed up by the factor of 1000. Error levels
			in the purely idiosyncratic case are
			measured in terms of total variation distance between approximations and
			the binomial distribution with parameters $(10,000,0.05)$ which arises as
			the independent sum of all Bernoulli random variables.
			Error levels in the purely non-idiosyncratic case are
			measured in terms of total variation distance between approximations and
			the mixed binomial distribution where for the ECRP model we use Poisson approximation to get an upper bound.
			the total variation between those distributions is $0.0159$ in our simulation and, thus,
			dominates the Poisson approximation in terms of total variation.
			Results are summarised in Table \ref{table_est}.
		
		\begin{table}[ht]
	\begin{center}\footnotesize{
		\caption{uantiles, execution times (speed) and total variation distance (accuracy) of Monte Carlo with Bernoulli deaths and $50,000$ simulations, as well as the \emph{extended CreditRisk$^+$} (ECRP) model with Poisson deaths, given~a~simple~portfolio.} 
		\label{table_est}
		\begin{tabular}{r|rrrrr|rr}
		\cline{2-6}\rule[-6pt]{0pt}{18pt}
			& \multicolumn{5}{c|}{quantiles}  & &  \\ 
			\cline{2-8}\rule[-6pt]{0pt}{18pt}
			& 1\%  & 10\% & 50\%& 90\%  & 99\%  & speed & accuracy \\ 
			\hline\rule[10pt]{0pt}{0pt}
			Bernoulli (MC), $w_{i,0}=1$ & $450$ & $472$ & $500$ & $528$ &  $552$ & $22.99$ sec. & $0.0187$\\\rule[-5pt]{0pt}{0pt}
			Poisson (ECRP), $w_{i,0}=1$ & $449$ & $471$ & $500$ & $529$ & $553$ & $0.01$ sec. & $0.0125$\\
			\hline\rule[10pt]{0pt}{0pt}
			Bernoulli (MC), $w_{i,1}=1$ & $202$ & $310$ & $483$ & $711$ & $936$ & $23.07$ sec. & $0.0489$\\\rule[-5pt]{0pt}{0pt}
			Poisson (ECRP), $w_{i,1}=1$ & $204$ & $309$ & $483$ & $712$ & $944$ & $0.02$ sec. &$ \leq 0.0500$\\\hline
		\end{tabular}}
	\end{center}
\end{table}

\subsection{Application I: Mortality Risk, Longevity Risk and Solvency II Application}\label{sec:SII}
In light of the previous section, life tables can be projected into the future and, thus, it is straightforward to derive best estimate  liabilities (BEL) of annuities and life insurance contracts. The possibility that death probabilities differ from an expected curve, i.e.,~estimated parameters do no longer reflect the best estimate and have to  be changed, contributes to mortality or longevity risk, when risk is measured over a one year time horizon as in Solvency II and the duration of in-force insurance contracts exceeds this time horizon. In our model, this risk can be captured by considering various MCMC samples $(\hat{\theta}^h)_{h= 1,\dots, m}$ (indexed by superscript $h$) of parameters $\theta=(\alpha,\beta,\zeta,\eta,\gamma)$ for death probabilities, yielding distributions of BELs. For example, taking $D(T,T+t)$ as the discount curve from time $T+t$ back to $T$ and choosing an MCMC sample $\hat{\theta}^h$ of parameters to calculate death probabilities
$q^h_{a,g}(T)$ and survival probabilities ${p}^h_{a,g}(T)$ at age $a$ with gender $g$, the BEL for a term life insurance contract which pays $1$ unit at the end of the year of death within the contract term of $d$ years is given by
\begin{equation}\label{BEL_life_insurance}
	A_{a,g}^T\big(\hat{\theta}^h\big)=D(T,T+1) q^h_{a,g}(T)+\sum_{t=1}^d D(T,T+t+1)\cdot {}_{t} {p}^h_{a,g}(T)  q^h_{a+t,g}(T+t)\,.
\end{equation}

In a next step, this approach can be used as a building block for (partial) internal models to calculate basic solvency capital requirements (BSCR) for biometric underwriting risk under Solvency~II, as~illustrated in the following example.

Consider an insurance portfolio at time $0$ with $E\in\mathbb{N}$ whole life insurance policies with lump sum payments $C_i>0$, for $i=1,\dots,E$, upon death at the end of the year. Assume that all assets are invested in an EU government bond (risk free under the standard model of the Solvency II directive) with maturity $1$, nominal $A_0$  and coupon rate $c>-1$. Furthermore, assume that we are only considering mortality risk and ignore profit sharing, lapse, costs, reinsurance, deferred taxes, other assets and other liabilities, as well as the risk~margin. Note~that in this case, basic own funds, denoted by $\mathrm{BOF}_t$, are given by market value of assets minus BEL at time $t$, respectively. Then, the BSCR at time $0$ is given by the 99.5\% quantile of the change in basic own funds over the period $[0,1]$, denoted~by $\Delta\mathrm{BOF}_1$, which~can be derived by, see (\ref{BEL_life_insurance}),
\begin{equation}\label{BOF}
\begin{array}{l}
	\begin{aligned}
	\Delta\mathrm{BOF}_1&=\mathrm{BOF}_0- D(0,1)\mathrm{BOF}_1=
		A_0\big(1- D(0,1)(1+c)\big)-\sum_{i=1}^E C_i A_{a,g}^{0}\big(\hat{\theta}\big)\\
		&+\frac{D(0,1)}{m}\sum_{h=1}^m \bigg(\sum_{i=1}^E C_i  A_{a+1,g}^{1}\big(\hat{\theta}^h\big)+\sum_{i=1}^E\sum_{j=1}^{N_i^h}C_i \big(1-A_{a+1,g}^{1}\big(\hat{\theta}^h\big)\big)\bigg)\,.
	\end{aligned}
	\end{array}
\end{equation}
where $\hat \theta:=\frac{1}{m}\sum_{h=1}^m \hat{\theta}^h$ and where $N^h_1,\dots,N^h_E$ are independent and Poisson distributed with $\mathbb{E}[{N_i^h}]=q_{a_i,g_i}^h(0)$ with policyholder $i$ belonging to age group $a_i$ and of gender $g_i$. The distribution of the last sum above can be derived efficiently by Panjer recursion. This example does not require a consideration of market risk and it nicely illustrates how mortality risk splits into a part associated with statistical fluctuation (experience variance: Panjer recursion) and into a part with long-term impact (change in assumptions: MCMC). Note that by mixing $N_i$ with common stochastic risk factors, we may include other biometric risks such as morbidity.

Consider a portfolio with 100  males and females at each age between 20 and 60 years, each having a 40-year term life insurance, issued in 2014, which provides a lump sum payment between 10,000 and 200,000 (randomly chosen for each policyholder) if death occurs within these 40 years. Using~MCMC samples and estimates based on the Austrian data
from 1965 to 2014 as given in the previous section, we may derive the change in basic own funds from 2014 to 2015 by (\ref{BOF}) using the extended CreditRisk$^+$ algorithm. The 99.5\% quantile of change in BOFs, i.e.,~the SCR, is lying slightly above one million. If~we did not consider parameter risk in the form of MCMC samples, the SCR would decrease by roughly 33\%.

\subsection{Application II: Impact of Certain Health Scenarios in Portfolios}
Analysis of certain health scenarios and their impact on portfolio P\&L distributions is straightforward
As an explanatory example, assume $m=1600$ policyholders which distribute uniformly over all age categories and genders, i.e., each age
category contains 100 policyholders with corresponding death probabilities, as well as weights as previously estimated and
forecasted for 2012.
Annuities $Y_i$ for all $i\in\{1,\dots,E\}$ are paid annually and take deterministic values in $\{11,\dots,20\}$ such that ten policyholders in
each age and gender category share equally high payments.
We now analyse the effect on the total amount of annuity payments in the next period under the scenario,
indexed by \lq scen\rq , that deaths due to neoplasms are reduced by 25\% in 2012 over all ages. In that case, we can estimate the realisation of risk factor
for neoplasms, see (\ref{lambda_MAP}), which takes an estimated value of $0.7991$.
Running the ECRP model with this risk factor realisation being fixed, we end up
with a loss distribution $L^\mathrm{scen}$ where deaths due to neoplasms have decreased.
Figure \ref{fig:scen} then shows probability distributions of traditional loss $L$
without scenario, as well as of scenario loss  $L^{\mathrm{neo}}$ with corresponding 95\% and 99\%
quantiles.
We observe that a reduction of 25\% in cancer crude death rates leads to a
remarkable shift in quantiles of the loss distribution as fewer people die and, thus,
more annuity payments have to be made.
\begin{figure}[ht]
	\centering
		\includegraphics[width=0.62\textwidth,trim=0 0cm 0 0cm,clip]{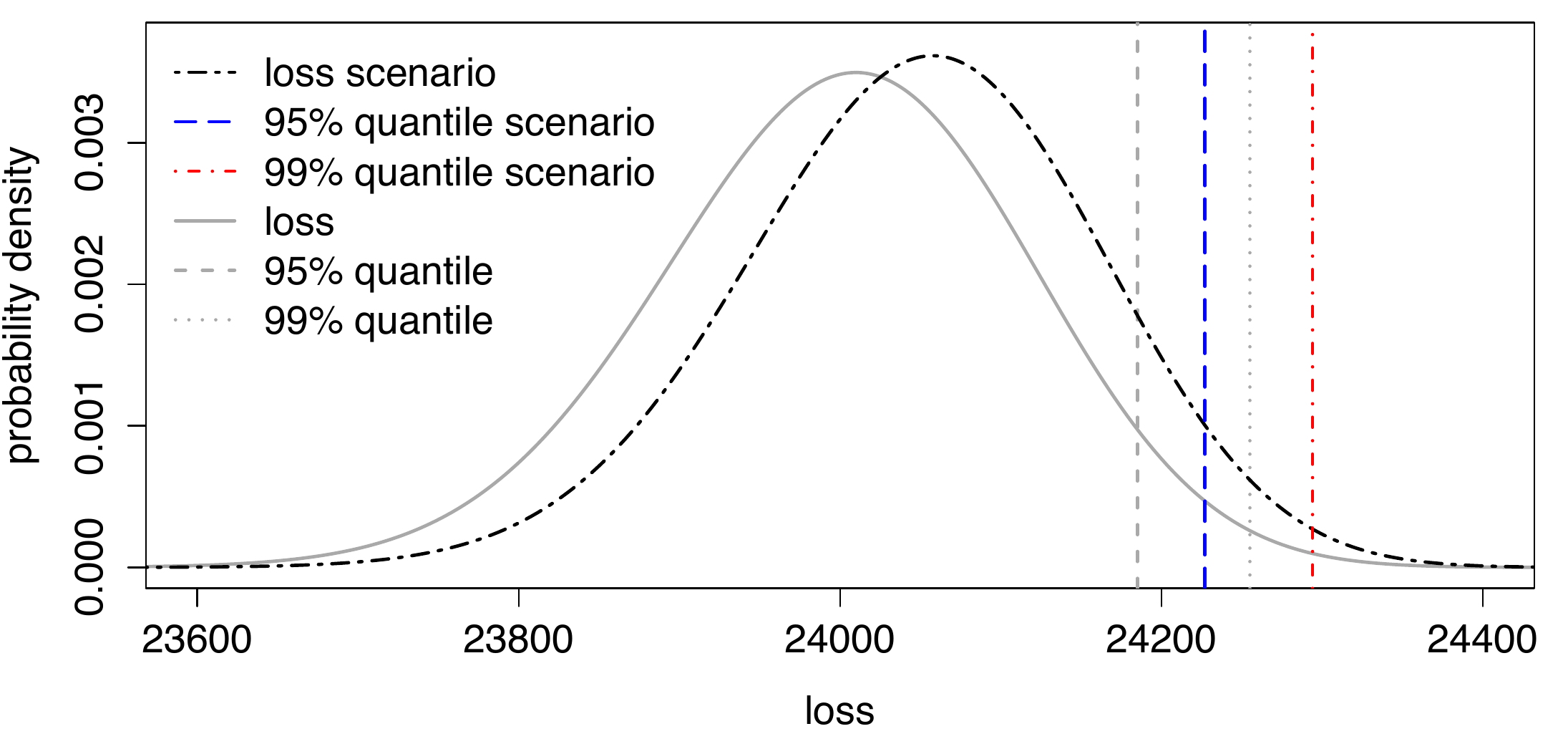}
		\caption[Example Australia: Loss distribution for simple portfolio with scenario.]{Loss distributions of\/ $L$ and\/ $L^{\mathrm{scen}}$ with\/ 95 and\/ 99\% quantiles.}
		\label{fig:scen}
	
\end{figure}

\subsection{Application III: Forecasting Central Death Rates and Comparison With the Lee--Carter Model}\label{sec:LC}
We  can compare out-of-sample
forecasts of death rates from our proposed model to forecasts obtained by other mortality models. Here, we choose the traditional Lee--Carter model as a proxy as our proposed model is conceptionally based on a similar idea.
We make the simplifying assumption of a constant population for out-of-sample time points.

Using the ECRP model it is straight-forward to {forecast central death rates}
and to give corresponding confidence intervals via setting $Y_j(t):=1$.
Then, for an estimate $\hat{{\theta}}$ of parameter vector
${\theta}$ run
the ECRP model with parameters
forecasted, see (\ref{eq:PDFamily}) and (\ref{eq:WeightFamily}).
We then obtain the distribution of the total number
of deaths $S_{a,g}(t)$ given $\hat{{\theta}}$ and, thus,
forecasted death rate $\widehat{m}_{a,g}(t)$ is given by
$\mathbb{P}\big(\widehat{m}_{a,g}(t)=N/{E_{a,g}(T)}\big)
	=\mathbb{P}(S_{a,g}(t)=N)$, for all $N\in\mathbb{N}_0$.
	
Uncertainty in the form of confidence intervals represent
statistical fluctuations, as well as random changes in risk factors. Additionally, using results obtained
by Markov chain Monte Carlo (MCMC) it is even possible to
incorporate parameter uncertainty into predictions.
To account for	 an increase in uncertainty for forecasts we suggest to assume increasing risk factor variances for forecasts, e.g., $\tilde\sigma^2_k(t)=\sigma^2_k (1+d(t-T))^2$ with $d\geq 0$.
A motivation for this approach with $k=1$ is the following: A~major source of uncertainty for forecasts lies in an
unexpected deviation from the estimated trend for death probabilities. We may therefore assume that rather than being deterministic, forecasted
values $m_{a,g}(t)$ are beta distributed (now denoted by $M_{a,g}(t)$) with $\mathbb{E}[{M_{a,g}(t)}]=m_{a,g}(t)$ and
variance $\sigma^2_{a,g}(t)$ which is increasing in time. Then, given independence amongst risk factor $\Lambda_1$ and
$M_{a,g}(t)$, we may assume that there exists a future point in time $t_0$ such that
\[
	\sigma^2_{a,g}(t_0)=\frac{m_{a,g}(t_0) (1- m_{a,g}(t_0))}{\sigma^{-2}_1+1}\,.
\]

In that case, $M_{a,g}(t_0) \Lambda_1$ is again gamma distributed with mean one and increased variance
$m_{a,g}(t_0)\sigma^2_1$ (instead of $m^2_{a,g}(t_0)\sigma^2_1$ for the deterministic case).
Henceforth, it seems reasonable to stay within the family of gamma distributions for forecasts and just adapt
variances over time. Of course, families for these variances for gamma distributions can be changed arbitrarily
and may be selected via classical information criteria.

Using	in-sample data, $d$ can be estimated via (\ref{MLE_likelihood}) with all other parameters being fixed. Using~Australian death and population data for the years 1963 to 1997 we estimate model parameters via MCMC in the ECRP model with one common stochastic risk factor having constant weight one. In~average, i.e., for various forecasting periods and starting at different dates, parameter $d$ takes the value $0.22$ in our example.
Using fixed trend parameters as above, and using the mean of 30,000 MCMC samples, we forecast death rates and corresponding confidence intervals out of sample
for the period 1998 to 2013. We can then compare these results to crude death rates within the stated period and
to forecasts obtained by the Lee--Carter model which is shown in Figure \ref{fig:PDforecast} for females aged $50$ to $54$~years.
We observe that crude death rates mostly fall in the 90\% confidence band for both procedures.
Moreover, Lee--Carter forecasts lead to wider spreads of quantiles in the future whilst the ECRP model suggests a more moderate increase in uncertainty. Taking various parameter samples from the MCMC chain and deriving quantiles for death rates, we can extract contributions of parameter
uncertainty in the ECRP model coming from posterior distributions of parameters.
\begin{figure}[ht]
	\centering
		\captionsetup{aboveskip=0.2\normalbaselineskip}
		\includegraphics[width=0.65\textwidth]{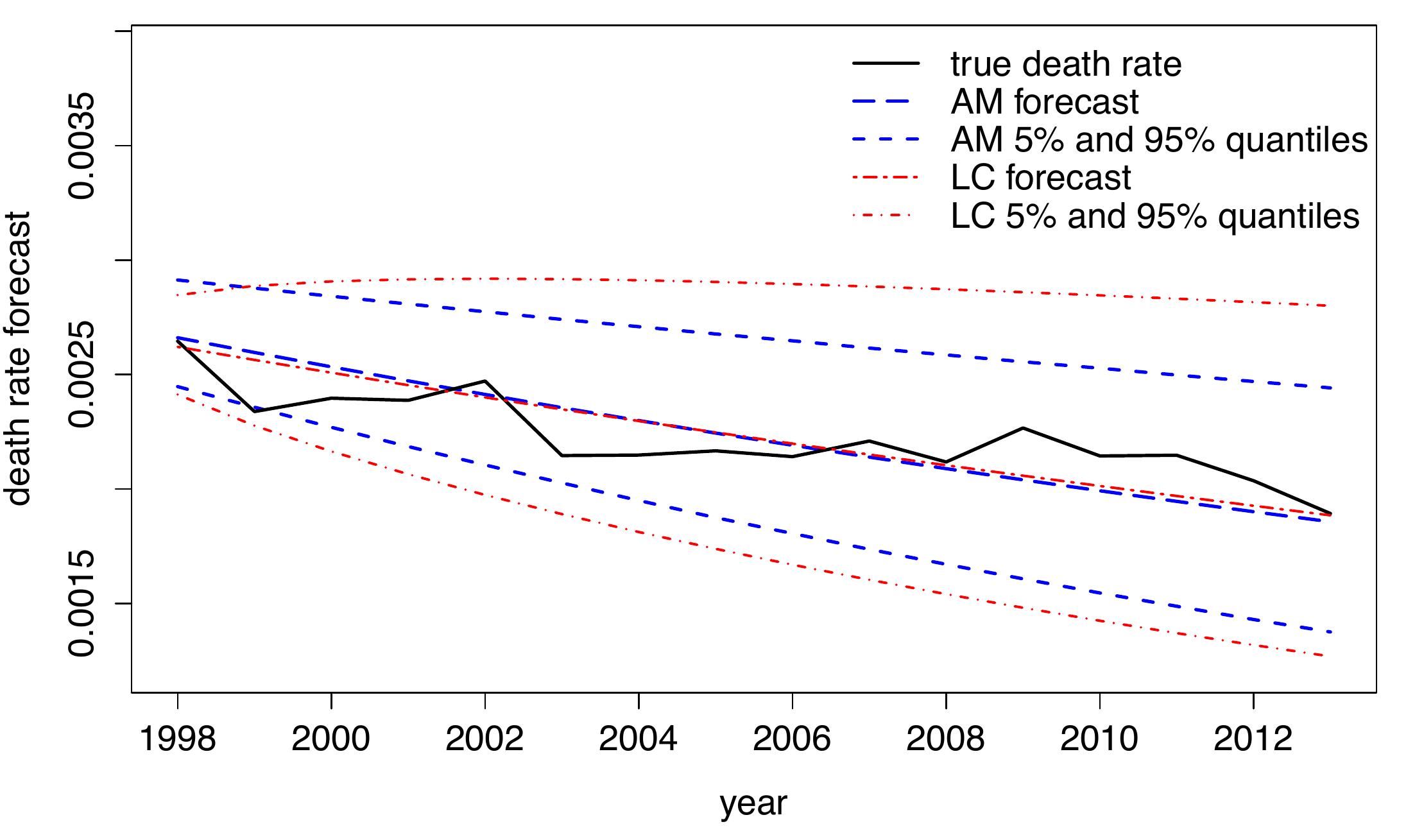}
		\caption[Example Australia: True versus forecasted death rates 2002--2011.]{Forecasted and true death rates using the ECRP model (AM) and the Lee--Carter model (LC) for females aged $50$ to $54$~years.}
		\label{fig:PDforecast}
	
\end{figure}

Within our approach to forecast death rates, it is now possible to derive contributions of various sources of risk. If we set $\delta=0$ we get forecasts where uncertainty solely comes from statistical fluctuations and random changes in risk factors. Using
$\delta = 0.22$ this adds the uncertainty increase associated with uncertainty for forecasts.
Finally, considering several MCMC samples this adds parameter risk. We observe that the contribution of statistical fluctuations and random changes in risk factors decreases from 63\% in 1998 to 20\% in 2013. Adding the increase in uncertainty for forecasts gives a roughly constant contribution of 72\% which implies that $\delta$ becomes the main driver of risk in the long term. On top of that, parameter uncertainty leaves a constant contribution of 28\%.

\subsection{Considered Risks}

Regulators often require {security margins} in life tables when modelling annuity or certain life insurance products and
portfolios to account
for different {sources of risk}, including trends, volatility~risk, model risk and parameter risk, ~\cite{AVOE_annuity_table} as well as \cite{DAV}.

In the ECRP model, {mortality trends}  are incorporated via families for death probabilities which are motivated by
the Lee--Carter model. It is straight forward to arbitrarily change parameter
families such that it fits the data as in the case when trends change fundamentally.
If other families for weights are used, one always has to check that they sum up to one over all death causes.
Note that for certain alternative parameter families, mean estimates obtained from Markov chain
Monte Carlo do not necessarily sum up to one anymore.
Changing model parameter families may also be necessary when using long-term
projections since long-term trends
are fundamentally different from short-term trends. Further estimation and testing procedures for trends in composite Poisson models in the context of
convertible bonds can be found in \cite{schmockConvBonds}.
Trends for weights are particularly interesting insofar as the model becomes
sensitive to the change in the vulnerability of policyholders to different death causes over time. Cross dependencies
over different death causes and different ages can~occur. \mbox{Such an effect}
can arise as a reduction in deaths of a particular cause can
lead to more deaths in another cause, several periods later, as people have to die at some
point. Furthermore, the~ECRP model captures unexpected, temporary deviations from a trend  with the variability introduced by common stochastic risk factors which effect all policyholders
according to weights simultaneously.

Assuming that the model choice is right and that estimated values are correct, life tables still
just give mean values of death probabilities over a whole population. Therefore,
in the case of German data it is suggested to add a non gender specific
due to legal reasons and it is set to 7.4\% to account for the risk of random fluctuations in deaths, approximately at a $95$\% quantile, see German Actuarial Association (DAV) \cite[section 4.1]{DAV2}.
In the ECRP model this risk is captured automatically by risk aggregation.
As a reference to the suggested security margin of $7.4$\% on death probabilities, we can use the same approach as given in
Section \ref{sec:LC} to estimate quantiles for death rates via
setting $Y_j=1$. These~quantiles then
correspond to statistical fluctuations around death probabilities.
We roughly
observe an~average deviation from death probability of 8.4\% for the 5\% quantile
and of 8.7\% for the 95\% quantile
of females aged 55 to 60 years in 2002, i.e., these values are in line with
a security margin of 7.4\%.

The risk of wrong parameter estimates, i.e.,~that realised death probabilities deviate from estimated values,
can be captured using MCMC
as described in Section \ref{sec:MCMC} where we sample from the joint posterior distributions of the estimators.
As our proposed extended CreditRisk$^+$ algorithm is numerically very efficient, we can easily run the ECRP model for several thousand samples from
the MCMC chain to derive
sensitivities of quantiles, for example. Note that parameter risk is closely linked to longevity risk. To cover the risk of fundamental changes in future death
probabilities, Section~\ref{sec:LC} provides an approach where future risk factor variances increase over time.

Modelling is usually a projection of a sophisticated real world problem on a relatively simple subspace which cannot cover all
facets and observations in the data.
Therefore, when applying the ECRP model to a portfolio of policyholders, we usually find {structural differences} to the
data which is used for estimation. There may also be a difference in mortality rates between individual companies or between portfolios within a company since
different types of insurance products attract different types of policyholders with a different individual risk profile.
In Germany, for these risks a minimal
security margin of ten\% is suggested, see \cite[section 2.4.2]{DAV}. Within the ECRP model, this~risk
can be addressed by using individual portfolio data instead of the whole population.
Estimates from the whole population or a different portfolio can be used
as prior distributions under MCMC which, in case of sparse data, makes estimation more stable. Another~possibility for introducing dependency amongst two portfolios is the introduction of a joint stochastic risk factor for both portfolios. In that case, estimation can be performed jointly with all remaining (except risk factors and their variances) parameters being individually estimated for both portfolios.
In contrast to the whole population, observed mortality rates
in insurance portfolios often show a completely different structure due to self-selection of policyholders. In particular,
for ages around 60, this effect is very~strong. In~Germany, a security margin for death probabilities of 15\% is suggested to cover
selection effects, see DAV \cite[section 4.2]{DAV2}.
In the literature, this~effect is also referred to as basis risk, \cite{basisRisk}. As already mentioned, instead of using a fixed security margin, this~issue can be tackled by using portfolio data with estimates from the whole population serving as prior distribution. Again, dependence amongst a portfolio and the whole population can be introduced by a~joint stochastic risk factor in the ECRP model.

Alternatively, in~\cite[section 4.7.1]{AVOE_annuity_table} it is suggested that all these risks are addressed by
adding a constant security margin on the trend.
This approach has the great conceptional advantage that the security margin is increasing over time and does
not diminish as in the case of direct security margins on death probabilities.

\subsection{Generalised and Alternative Models}\label{alt_models}
Up to now, we applied a simplified version of extended CreditRisk$^+$ to derive cumulative payments in annuity
portfolios. A major shortcoming in this approach is the limited possibility of modelling dependencies amongst policyholders and  death causes. In the most general form of extended CreditRisk$^+$  as described in \cite[section 6]{schmock},
it is possible to introduce
risk groups which enable us to model joint deaths of several policyholders and it is possible to model
dependencies amongst death causes.  Dependencies can take a linear
dependence structure combined with dependence scenarios to model negative correlations as well. Risk factors may then be identified with statistical variates
such as average blood pressure, average physical activity or the average of smoked cigarettes, etc.,
and not directly with death causes.
Moreover, for each policyholder individually,
the general model allows for losses  which depend on the underlying cause of death.
This gives scope to the possibility of modelling---possibly new---life insurance products with payoffs depending on the cause of death as, for example, in
the case of accidental death benefits. Including all extensions mentioned above, a similar algorithm may still be applied to derive loss distributions,  see~\cite[section 6.7]{schmock}.

Instead of using extended CreditRisk$^+$ to model
annuity portfolios, i.e., an approach based on Poisson mixtures, we can assume a similar {Bernoulli mixture model}.  In such a Bernoulli mixture
model, conditionally Poisson distributed deaths are simply replaced by conditionally Bernoulli distributed~deaths.
In general, explicit and efficient derivation of loss distributions in the case of Bernoulli mixture models is not possible anymore.
Thus, in this case,
one has to rely on other methods such as Monte Carlo.
Estimation of model parameters works similarly as discussed in Section \ref{sec:estimation}. Poisson~approximation suggests that loss distributions
derived from Bernoulli and Poisson mixture models are similar in terms of total variation distance if death probabilities are small.

\section{Model Validation and Model Selection}\label{validation}

In this section we propose several validation techniques in order to check whether the ECRP model fits the given data or not. Results for Australian data, see Section \ref{sec:real}, strongly suggest that the proposed model is suitable.
If any of the following validation approaches suggested misspecification in the model or if parameter estimation did not
seem to be accurate, one possibility to tackle these problems would be to reduce risk factors.

\subsection{Validation via Cross-Covariance}\label{sec:validationCross}
For the first procedure, we transform deaths $N_{a,g,k}(t)$ to $N'_{a,g,k}(t)$, see Section \ref{sec:MM}, such that this sequence has constant expectation and can thus be assumed to be i.i.d. Then, sample variances of transformed death counts, cumulated across age and gender groups, can be compared to MCMC confidence bounds from the model.
In the death-cause-example all observed sample variances of $N_{k}(t)$ lie within 5\%- and 95\%-quantiles.

\subsection{Validation via Independence}\label{sec:validation_ind}

Number of deaths for different
death causes are independent within the ECRP model as independent risk factors are assumed.
Thus, for all $a,a'\in\{1,\dots,A\}$ and $g,g'\in\{\mathrm{f},\mathrm{m}\}$, as well as
$k,k'\in\{0,\dots,K\}$ with $k\neq k'$ and $t\in U$, we have
$\text{Cov}(N_{a,g,k}(t), N_{a',g',k'}(t))=0$.
Again, transform the data as above and subsequently normalise the transformed data, given $\text{Var}\big(N'_{a,g,k}(t)|\Lambda_k(t)\big)>0$ a.s.,
as follows:
\[
	N^*_{a,g,k}(t)
	:=\frac{N'_{a,g,k}(t)-\mathbb{E}[{N'_{a,g,k}(t)}|{\Lambda_k(t)}]}
	{\sqrt{ \text{Var}({N'_{a,g,k}(t)}|{\Lambda_k(t)})}}
	=\frac{N'_{a,g,k}(t)-E_{a,g}m_{a,g} w_{a,g,k}\Lambda_k(t)}
	{\sqrt{ E_{a,g}m_{a,g} w_{a,g,k}\Lambda_k(t)}}\,.
\]

Using the conditional
central limit theorem as in \cite{CCLT}, we~have $N^*_{a,g,k}(t)\rightarrow N(0,1)$ in distribution as $E_{a,g}(t)\rightarrow\infty$
where $N(0,1)$ denotes the standard normal distribution. Thus,
normalised death counts $n^*_{a,g,k}(t)$ are given by
\[
	n^*_{a,g,k}(t)
	=\frac{n'_{a,g,k}(t)-E_{a,g}\hat{m}_{a,g} \hat{w}_{a,g,k}\hat{\lambda}_k(t)}
	{\sqrt{ E_{a,g}\hat{m}_{a,g} \hat{w}_{a,g,k}\hat{\lambda}_k(t)}}\,.
\]
with $\hat{\lambda}_0(t):=1$. Then, assuming that each pair $(N^*_{a,g,k}(t),N^*_{a',g',k'}(t))$,
for $a,a'\in\{1,\dots,A\}$ and  $g,g'\in\{\mathrm{f},\mathrm{m}\}$, as well as
$k,k'\in\{0,\dots,K\}$ with $k\neq k'$ and $t\in U$,
has a joint normal distribution with some correlation coefficient $\rho$ and standard normal marginals, we may
derive the sample correlation coefficient
\[
	R_{a,g,a',g',k,k'}:=
		\frac{\sum_{t=1}^T (N^*_{a,g,k}(t)- \overline{N}^*_{a,g,k})
	 	(N^*_{a',g',k'}(t)- \overline{N}^*_{a',g',k'})}
		{\sqrt{\sum_{t=1}^T (N^*_{a,g,k}(t)- \overline{N}^*_{a,g,k})^2
	 	\sum_{t=1}^T (N^*_{a',g',k'}(t)- \overline{N}^*_{a',g',k'})^2}}\,,
\]
 where $ \overline{N}^*_{a,g,k}:=\frac{1}{T}\sum_{s=1}^T N^*_{a,g,k}(s)$.
Then, the test of the null hypothesis $\rho=0$  against the alternative hypothesis $\rho\neq 0$ rejects the null hypothesis
at an $\delta$-percent level, see \cite[section 5.13]{statistic}, when
\begin{equation}\label{ind}
	\frac{|R_{a,g,a',g',k,k'}|}{\sqrt{(1-R_{a,g,a',g',k,k'}^2)/
	(T-2)}}
	>K_{\delta,T}\,,
\end{equation}
with $K_{\delta,T}$ such that $\int_{K_{\delta,T}}^\infty
t_{T-2}(y)\, dy=\delta/2$ where $t_{T-2}$ denotes the density
of a $t$-distribution with $(T-2)$ degrees of freedom.

Applying this validation procedure on Australian data with ten death causes shows that 88.9\%
	of all independence tests, see (\ref{ind}), are accepted at a 5\% significance level.
	Thus, we may assume that the ECRP model fits the data suitably with respect to independence amongst
	death counts due to different~causes.

\subsection{Validation via Serial Correlation}\label{sec:serial}

Using the same data transformation and normalisation as in Section \ref{sec:validation_ind}, we may assume that random variables $(N^*_{a,g,k}(t))_{t\in U}$ are identically and
standard normally distributed.
Then, we can check for serial dependence and autocorrelation in the data such as causalities between a reduction in deaths due to certain death causes and a possibly lagged increase
in different ones. Note that we already remove a lot of dependence via time-dependent weights and
death probabilities. Such serial effects are, for example, visible in the case of mental and
behavioural disorders and circulatory diseases.

Many tests are available most of which assume an autoregressive model with normal errors such as
the Breusch--Godfrey test, see~\cite{breusch}. For the Breusch--Godfrey test a linear model is fitted to the data where the residuals
are assumed to follow an autoregressive process of length $p\in\mathbb{N}$. Then,  $(T-p)R^2$ asymptotically
follows a $\chi^2$ distribution with $p$ degrees of freedom under the null hypothesis that there is no autocorrelation. In \lq R\rq, an implementation of the Breusch--Godfrey is available
within the function \textsf{bgtest} in the \lq lmtest\rq~package, see \cite{lmtest}.

Applying this validation procedure to Australian data given in Section \ref{sec:real},
	the null hypothesis, i.e., that there is no serial correlation of order $1,2,\dots, 10$, is not rejected at a 5\% level
	 in 93.8\% of all cases.
	Again, this is an indicator that the ECRP model with trends for weights and death probabilities fits the data suitably

\subsection{Validation via Risk Factor Realisations}

In the ECRP model, risk factors ${\Lambda}$ are assumed to be independent and identically gamma
distributed with mean one and variance $\sigma^2_k$. Based on these assumptions,
we can use estimates for risk factor realisations
${\lambda}$ to judge whether the ECRP model adequately fits the data.
These estimates can either be obtained via MCMC based on the maximum a posteriori setting or by
Equations (\ref{lambda_MAP}) or~(\ref{MAPappr_lambda}).

For each $k\in\{1,\dots,K\}$, we may check whether estimates
$\hat{\lambda}_k(1),\dots,\hat{\lambda}_k(T)$ suggest a rejection of the null hypothesis that
they are sampled from a gamma distribution with mean one and variance $\sigma^2_k$.
The classical way is to use the Kolmogorov--Smirnov test, see e.g. \cite[section~6.13]{statistic} and the references therein.
In \lq R\rq~an implementation of this test is provided by the \textsf{ks.test} function, see \cite{stats}. The null hypotheses is rejected as soon as the test statistic
$\sup_{x\in\mathbb{R}}|F_T(x)-F(x)|$ exceeds the corresponding critical value where
$F_T$ denotes the empirical distribution function of samples $\hat{\lambda}_k(1),\dots,\hat{\lambda}_k(T)$ and where  $F$ denotes the gamma distribution function with mean one and variance~$\sigma^2_k$.

	Testing whether risk factor realisations are sampled from a gamma distribution via
	the Kolmogorov--Smirnov test as described above gives acceptance of
	the null hypothesis for all ten risk factors on all suitable levels of significance.

\subsection{Model Selection}

For choosing a suitable family for mortality trends, information criteria such as AIC, BIC, or~DIC can be applied straight away. The decision how many risk factors to use cannot be answered by traditional information criteria since a reduction in risk factors leads to a different data structure. It~also depends on the ultimate goal. For example, if the development of all death causes is of interest, then a reduction of risk factors is not wanted. On the contrary, in the context of annuity portfolios several risk factors may be merged to one risk factor as their contributions to the risk of the total portfolio are~small.

\section[Conclusion]{Conclusions}
 We introduce an additive stochastic mortality model which is closely related to classical approaches such as the Lee--Carter model but allows for joint modelling of underlying death causes and improves models using disaggregated death cause forecasts. Model parameters can be jointly estimated using MCMC based on publicly available data. We give a link to extended CreditRisk$^+$ which provides a useful actuarial tool with numerous portfolio applications such as P\&L derivation in annuity and life insurance portfolios or (partial) internal model applications. Yet, there exists a fast and numerically stable algorithm to derive loss distributions exactly, instead of Monte Carlo, even for large portfolios. Our proposed model directly incorporates various sources of risk including trends,
			longevity, mortality risk, statistical volatility and
		estimation risk. In particular, it is possible to quantify the risk of statistical fluctuations within the next
period (experience variance) and parameter uncertainty over a longer time horizon (change in assumptions). Compared to the Lee--Carter model, we have a more flexible framework and can directly extract several sources of uncertainty. Straightforward model validation techniques are available.

\vspace{6pt}

\bibliographystyle{amsplain}

\renewcommand\bibname{References}

\end{document}